\documentclass[acmsmall,screen]{acmart}

\usepackage{proof}
\usepackage[utf8]{inputenc}
\usepackage{hyperref}
\NewCommandCopy{\oldGamma}{\Gamma}\renewcommand{\Gamma}{{\mathit{\oldGamma}}}
\NewCommandCopy{\oldDelta}{\Delta}\renewcommand{\Delta}{{\mathit{\oldDelta}}}
\NewCommandCopy{\oldTheta}{\Theta}\renewcommand{\Theta}{{\mathit{\oldTheta}}}
\newcommand{\e}[1]{{\color{brickred}{#1}}}
\newcommand{\blank}{\mathord{\hspace{1pt}\text{--}\hspace{1pt}}} %from the book
\newcommand{\Set}{\mathsf{Set}}
\newcommand{\C}{\mathcal{C}}
\newcommand{\ra}{\rightarrow}
\newcommand{\Ra}{\Rightarrow}
\newcommand{\id}{\mathsf{id}}

\renewcommand{\o}{0}
\renewcommand{\i}{1}
\renewcommand{\oe}{\e{0}}
\newcommand{\ie}{\e{1}}
\newcommand{\ke}{\e{k}}
\newcommand{\sym}{\mathsf{sym}}

\newcommand{\Con}{\mathsf{Con}}
\newcommand{\Sub}{\mathsf{Sub}}
\newcommand{\Ty}{\mathsf{Ty}}
\newcommand{\Tm}{\mathsf{Tm}}
\newcommand{\p}{\mathsf{p}}
\newcommand{\q}{\mathsf{q}}
\newcommand{\K}{\mathsf{K}}
\newcommand{\y}{\mathsf{y}}
\newcommand{\yl}{\mathsf{yl}}
\newcommand{\lam}{\mathsf{lam}}
\newcommand{\app}{\mathsf{app}}
\newcommand{\oldapp}{\mathbin{\$}}
\newcommand{\U}{\mathsf{U}}
\newcommand{\El}{\mathsf{El}}
\newcommand{\Eq}{\mathsf{Eq}}
\renewcommand{\c}{\mathsf{c}}

\renewcommand{\P}{\mathsf{P}}

\newcommand{\ext}{\mathop{\triangleright}}
\newcommand{\refl}{\mathsf{refl}}
\newcommand{\rel}{\mathsf{rel}}

\newcommand{\Id}{\mathsf{Id}}

\newcommand{\R}{\mathsf{R}}
\renewcommand{\Re}{\e{\mathsf{R}}}

\renewcommand{\tt}{\mathsf{tt}}
\newcommand{\suc}{\mathsf{suc}}
\renewcommand{\S}{\mathsf{S}}
\newcommand{\Se}{\e{\mathsf{S}}}
\newcommand{\Bool}{\mathsf{Bool}}
\newcommand{\true}{\mathsf{true}}
\newcommand{\false}{\mathsf{false}}
\newcommand{\ite}{\mathsf{ite}}
\newcommand{\foralld}{\mathord{{\forall}{\hspace{-0.15em}}\mathsf{d}}}
\newcommand{\ap}{\mathsf{ap}}
\newcommand{\apd}{\mathsf{apd}}
\newcommand{\unspan}{\mathsf{unspan}}
\newcommand{\unspane}{\e{\mathsf{unspan}}}
\renewcommand{\d}[1]{\mathord{#1{\hspace{-0.1em}}\mathsf{d}}}
\newcommand{\mkpi}{\mathsf{mk}{\forall}{\Pi}}
\newcommand{\mkpie}{\e{\mathsf{mk}{\forall}{\Pi}}}
\newcommand{\foralle}{\mathord{\e{\forall}}}
\newcommand{\cur}{\urcorner}
\newcommand{\cul}{\ulcorner}
\newcommand{\inl}{\mathsf{inl}}
\newcommand{\inr}{\mathsf{inr}}
\newcommand{\PSh}{\mathsf{PSh}}
\newcommand{\G}{\mathsf{G}}
\newcommand{\Syn}{\mathsf{Syn}}
\renewcommand{\ll}{\cul}
\newcommand{\rr}{\cur}

\usepackage{tikz}
\usetikzlibrary{arrows,cd}
\usepackage{bbm}

\definecolor{brickred}{rgb}{0.8, 0.25, 0.33}

\setcopyright{rightsretained}
\acmDOI{10.1145/3632920}
\acmYear{2024}
\copyrightyear{2024}
\acmSubmissionID{popl24main-p508-p}
\acmJournal{PACMPL}
\acmVolume{8}
\acmNumber{POPL}
\acmArticle{78}
\acmMonth{1}
\received{2023-07-11}
\received[accepted]{2023-11-07}

\begin{document}

\newtheorem{problem}[theorem]{Problem}
\theoremstyle{remark}
\newtheorem{construction}[theorem]{Construction}

\title{Internal Parametricity, without an Interval}

\author{Thorsten Altenkirch}
\orcid{}
\affiliation{%
  \institution{University of Nottingham}
  \city{}
  \country{UK}
}
\email{thorsten.altenkirch@nottingham.ac.uk}

\author{Yorgo Chamoun}
\orcid{}
\affiliation{%
  \institution{École Polytechnique}
  \city{}
  \country{France}
}
\email{yorgo.chamoun@polytechnique.edu}

\author{Ambrus Kaposi}
\orcid{}
\affiliation{%
  \institution{Eötvös Loránd University}
  \city{}
  \country{Hungary}
}
\email{akaposi@inf.elte.hu}

\author{Michael Shulman}
\orcid{}
\affiliation{%
  \institution{University of San Diego}
  \city{}
  \country{USA}
}
\email{shulman@sandiego.edu}

\begin{abstract}
  Parametricity is a property of the syntax of type theory implying,
  e.g., that there is only one function having the type of the
  polymorphic identity function. Parametricity is usually proven
  externally, and does not hold internally. Internalising it is
  difficult because once there is a term witnessing parametricity, it
  also has to be parametric itself and this results in the appearance
  of higher dimensional cubes. In previous theories with internal
  parametricity, either an explicit syntax for higher cubes is present
  or the theory is extended with a new sort for the interval. In this
  paper we present a type theory with internal parametricity which is
  a simple extension of Martin-Löf type theory: there are a few new
  type formers, term formers and equations. Geometry is not explicit
  in this syntax, but emergent: the new operations and equations only
  refer to objects up to dimension 3. We show that this theory is
  modelled by presheaves over the BCH cube category. Fibrancy
  conditions are not needed because we use span-based rather than
  relational parametricity. We define a gluing model for this theory
  implying that external parametricity and canonicity hold. The theory
  can be seen as a special case of a new kind of modal type theory,
  and it is the simplest setting in which the computational properties
  of higher observational type theory can be demonstrated.
\end{abstract}




\maketitle

\section{Introduction}

Parametricity was introduced by Reynolds \cite{DBLP:conf/ifip/Reynolds83} as a theory of
representation-indep\-endence for the polymorphic lambda calculus. The
idea is that a polymorphic function has to work uniformly on all
types, i.e., it cannot inspect its type arguments, and thus for example
there are zero, one and two terms of types $\forall a.a$, $\forall a.a
\ra a$ and $\forall a.a\ra a\ra a$, respectively. This intuition is
formalised by relation-preservation: each type is equipped with a
relation (called logical relation), and one can prove by induction on
the syntax that every term respects the relation corresponding to its
type (called the fundamental lemma). Dependent types are expressive
enough that they can formulate their own parametricity relations. This
was used by \cite{DBLP:conf/icfp/BernardyJP10} to define a parametricity
translation for type theory. We describe this translation below.
We expect that the reader is familiar with the syntax of type theory.

\paragraph{The external parametricity translation}

By mutual induction on syntactic contexts, substitutions, types and
terms, we define the following $\blank^\P$ operations. On contexts, we
further define operations $\o$ and $\i$ and we write $k$ when we mean
either. On substitutions we mutually prove an equation.
\[\arraycolsep=1.4pt
\infer{\begin{array}{l l}\Gamma^\P & : \Con \\ k_\Gamma & :\Sub\,\Gamma^\P\,\Gamma \end{array}}{\Gamma : \Con} \hspace{1.4em}
\infer{\begin{array}{l}\sigma^\P : \Sub\,\Delta^\P\,\Gamma^\P \\ \sigma\circ k_\Delta = k_\Gamma\circ\sigma^\P \\ \end{array}}{\sigma:\Sub\,\Delta\,\Gamma}\hspace{1.4em}
\infer{\begin{array}{l}A^\P : \Ty\,(\Gamma^\P,A[\o_\Gamma],A[\i_\Gamma]) \\ \phantom{k_\Gamma} \end{array}}{A : \Ty\,\Gamma} \hspace{1.4em}
\infer{\begin{array}{l}t^\P : \Tm\,\Gamma^\P\,(A^\P[t[\o_\Gamma], t[\i_\Gamma]]) \\ \phantom{k_\Gamma}\end{array}}{t : \Tm\,\Gamma\,A}
\]
The main point of this translation is to compute the logical relation
$A^\P$ from a type $A$ (the third operation). For a closed type $A :
\Ty\,\diamond$, we obtain a homogeneous binary relation $A^\P :
\Ty(\diamond,A,A)$, that is, a type depending on two variables of type
$A$. For an $A$ in a non-empty context, $A^\P$ is a heterogeneous
relation which depends on the operations for contexts. Because types
can include terms, we need to define $\blank^\P$ mutually on terms and
substitutions as well. We explain how these operations are defined for each
sort.
\begin{itemize}
\item $\Gamma^\P$ is a context that contains two copies of each type
  in $\Gamma$ together with witnesses of their relatedness. The empty
  context $\diamond$ stays empty. For a context ending with $A$, we
  obtain two copies of $A$ which are substituted by $\o$ and $\i$,
  respectively; finally we have a witness of relatedness. The
  projections $\o$ and $\i$ return the $x_\o$ and $x_\i$
  components, respectively.
\begin{alignat*}{10}
  & \diamond^\P && {}:= \diamond \hspace{7em} && (\Gamma,x:A)^\P && {}:= \Gamma^\P,x_0:A[\o_\Gamma],x_1:A[\i_\Gamma],x_2:A^\P[x_0,x_1] \\
  & k_\diamond && {}:= \epsilon && k_{\Gamma,x:A} && {}:= (k_\Gamma,x\mapsto x_k)
\end{alignat*}
\item A substitution is a list of terms, each variable $x$ in the
  codomain context is mapped to some term $t$ which we denote
  $x\mapsto t$. On the empty substitution $\blank^\P$ is the identity,
  on a substitution into an extended context
  $(\sigma,x\mapsto t) : \Sub\,\Delta\,(\Gamma,x:A)$, it is defined
  pointwise.
\begin{alignat*}{10}
  & \epsilon^\P && {}:= \epsilon \hspace{7em} && (\sigma,x\mapsto t)^\P && {}:= (\sigma^\P,x_0\mapsto t[\o_\Delta], x_1\mapsto t[\i_\Delta], x_2\mapsto t^\P)
\end{alignat*}
\item $A^\P$ is a heterogeneous relation between two different copies
  of $A$, the dependencies of which are given by $_0$ and $_1$
  components, respectively. For example,
  \[
  (\diamond,x:A,y:B)^\P = \diamond,x_0:A,x_1:A,x_2:A^\P[x_0,x_1],y_0:B[x_0],y_1:B[x_1],y_2:B^\P[y_0,y_1].
  \]
  It is defined separately for each type $A$. On the universe,
  $\blank^\P$ returns the relation space. On $\El$, it returns the
  type of witnesses of the relation using function application
  $\blank\oldapp\blank$. Both $\U$ and $\El\,a$ are types, $a$ is a
  term of type $\U$.
  \[
  \U^\P[a_0,a_1] := \El\,a_0\Ra\El\,a_1\Ra\U \hspace{5em} (\El\,a)^\P[x_0,x_1] := \El\,(a^\P\oldapp x_0\oldapp x_1)
  \]
\item The term $t^\P$ says that $t$ respects logical relations: if the
  relations hold for every dependency in the context, the relation
  $A^\P$ also holds for the two copies of $t$, depending on the
  respective copies of $\Gamma$. Sometimes $t^\P$ is called the
  fundamental lemma for the term $t$. Note that when we say
  ``relation'' we always mean proof-relevant relation (correspondence,
  or family of types).
\end{itemize}

\paragraph{The ``hello world'' example of parametricity}

A term in $\Tm\,(\diamond,x:\U,y:\El\,x)\,(\El\,x)$ can only contain
two free variables, $x$ and $y$. Using the translation $\blank^\P$, we
show that for any such term $t$ and closed terms $b$ and $u$, the
substituted term $t[b,u]$ is equal to $u$. This is one way to
formalise that $\forall a.a\ra a$ has only one element. In fact, the
unary version of the translation is enough, so for this example we
restrict ourself to the $k=\o$ case and omit the $_1$ components
(alternatively, we could fill the $_\i$ components using dummy $\top$
arguments). Now $t^\P$ says that if there is a code for a type $x_0$,
a predicate on elements of this type, and an element $y_0$ for which
the predicate holds, then the predicate will also hold for
$t[x_0,y_0]$.
\[
t^\P : \Tm\,\big(\diamond,x_0:\U,x_2:\El\,x_0\Ra\U,y_0:\El\,x_0,y_2:\El\,(x_2\oldapp y_0)\big) \big(\El\,(x_2\oldapp (t[x_0,y_0]))\big)
\]
Given a closed type $b : \Tm\,\diamond\,\U$, a term $u :
\Tm\,\diamond\,(\El\,b)$, we define a predicate $Q\oldapp y :=
\Eq_{\El\,b}\,u\,y$ expressing equality to $u$. As $Q$ holds for $u$
by reflexivity, we obtain the following substituted term with the
desired type.
\[
t^\P[b,Q,u,\refl] : \Tm\,\diamond\,(\Eq_{\El\,b}\,u\,(t[b,u]))
\]

The operation $\blank^\P$ can be defined for most well-behaved type
theories as a syntactic translation or model construction.  However, it
only gives parametricity in the empty context. In the above example if
we replace the empty context $\diamond$ by an arbitrary $\Gamma$, we
only obtain
\[
\Tm\,\Gamma^\P\,(\Eq_{\El\,b}\,u\,(t[b,u])),
\]
which expresses the same equality, but it is only valid in a different
(larger) context $\Gamma^\P$ which includes the information that the
elements in $\Gamma$ are themselves parametric. Operations that are
internal to type theory (such as $\lambda$ or application $\oldapp$) do
not act on the full context, as can be seen by looking at their
inference rules. This is in contrast with $\blank^\P$ which takes a
term into a completely different context. This is why $\blank^\P$ is
called an \emph{external} parametricity translation.

\paragraph{Internalising parametricity}

Internal parametricity can be obtained by postulating a substitution
$\R_\Gamma : \Sub\,\Gamma\,\Gamma^\P$ for every context $\Gamma$
\cite{DBLP:conf/types/AltenkirchK15}. Now given a $t :
\Tm\,(\Gamma,x:\U,y:\El\,x)\,(\El\,x)$, and $b$, $u$, $Q$ as above we
have
\[
t^\P[\R_\Gamma,b,Q,u,\refl] : \Tm\,\Gamma\,(\Eq_{\El\,b}\,u\,(t[b,u])),
\]
thus we obtain the desired equality in the same context as $t$.

However there is no hope of being able to define the substitution
$\R_\Gamma$ by induction on the syntax. Type theory has
non-parametric models in which the above equality does not hold for any
$t$: e.g., models with excluded middle \cite{DBLP:conf/types/BooijELS16}
or a type-case operator \cite{DBLP:conf/cpp/BoulierPT17}.

So it is not a surprise that when trying to define $\R_\Gamma$ by
induction on the context $\Gamma$, we need to extend the syntax by a
new operator which we call $\rel$:
\[
\R_\diamond := \epsilon \hspace{4em}
\R_{\Gamma,x:A} := (\R_\Gamma, x_0\mapsto x, x_1\mapsto x, x_2\mapsto \rel\,x) \hspace{4em}
\infer{\rel\,a :\Tm\,\Gamma\,(A^\P[\R_\Gamma,a,a])}{a : \Tm\,\Gamma\,A}
\]
Once we introduce new terms (such as $\rel$), we have to say how
$\blank^\P$ acts on them, and it is not clear how to do this. A
solution is to turn the $\blank^\P$ operations into operators of the
syntax and their definitions into conversion rules. Such a $\blank^\P$
relation behaves like an identity type that is reflexive and a
congruence, but has no transport (it is sometimes called a Bridge
type). Before defining our syntax with a Bridge type, we take a
detour to understand how iterated usages of $\blank^\P$ behave.

\paragraph{Higher cubes in external parametricity}

$\Gamma^\P$ can be seen as a context of lines, $(\Gamma^\P)^\P$ as a
context of squares, $((\Gamma^\P)^\P)^\P$ as a context of
three-dimensional cubes, and so on. We illustrate this by computing
the contents of $(\diamond,x:A)$ after applying $\blank^\P$ to it twice.
\begin{alignat*}{10}
  & {(\diamond,x:A)^\P}^\P && {}={} && \rlap{$(\diamond,x_0:A,x_1:A,x_2:A^\P[x_0,x_1])^\P =$} \\
  & && && \big(\diamond,{} && x_{00}:A,{} && x_{01}:A,{} && x_{02}:A^\P[x_{00},x_{01}], \\
  & && &&             && x_{10}:A,{} && x_{11}:A,{} && x_{12}:A^\P[x_{10},x_{11}], \\
  & && &&             && x_{20}:A^\P[x_{00},x_{10}],{} && x_{21}:A^\P[x_{01},x_{11}],{} && x_{22}:{A^\P}^\P[x_{00},x_{01},x_{02},x_{10},x_{11},x_{12},x_{20},x_{21}]\big)
\end{alignat*}
The contexts $(\diamond,x:A)$, $(\diamond,x:A)^\P$,
${(\diamond,x:A)^\P}^\P$, ${{(\diamond,x:A)^\P}^\P}^\P$ can be depicted as follows. \\
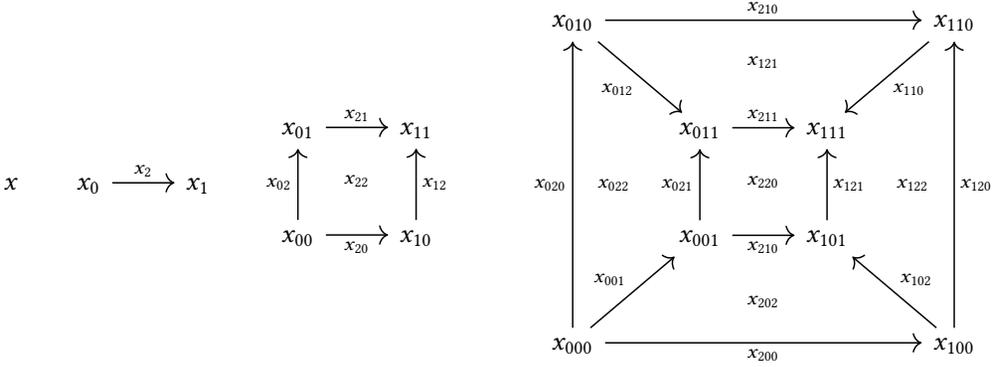
\begin{center}
    \begin{minipage}{0.25\textwidth}
        \begin{tikzcd}
            x\;\;\;\;\;\;\;\; x_0 \arrow[r, "x_2"] & x_1
        \end{tikzcd}
    \end{minipage}
    \begin{minipage}{0.25\textwidth}
        \begin{tikzcd}[row sep=large]
            x_{01} \arrow[r, "x_{21}"{name=r2}] & x_{11} \\
            x_{00} \arrow[u, "x_{02}"] \arrow[r, "x_{20}"'{name=r3}] & x_{10} \arrow[u, "x_{12}"']
        \arrow[from=r2, to=r3, phantom, "{\scriptstyle x_{22}}" description]
        \end{tikzcd}
    \end{minipage}
    \begin{minipage}{0.45\textwidth}
        \begin{tikzcd}[row sep=large]
        x_{010} \arrow[rrr, "x_{210}"{name=r1, above}]\arrow[dr, "x_{012}"'] &                                     &              &  x_{110} \arrow[dl, "x_{110}"]\\
                                & x_{011} \arrow[r, "x_{211}"{name=r2}]                                            & x_{111}      &  \\
                                & x_{001} \arrow[r, "x_{210}"{name=r3, below}]\arrow[u, "x_{021}"{name=u2}] & x_{101} \arrow[u, "x_{121}"{name=u3, right}]  &  \\
        x_{000} \arrow[rrr, "x_{200}"{name=r4, below}]\arrow[ur, "x_{001}"]\arrow[uuu, "x_{020}"{name=u1}]  &          &              & x_{100} \arrow[uuu, "x_{120}"{name=u4, right}] \arrow[ul, "x_{102}"']
        \arrow[from=r1, to=r2, phantom, "{\scriptstyle x_{121}}" description]
        \arrow[from=r2, to=r3, phantom, "{\scriptstyle x_{220}}" description]
        \arrow[from=r3, to=r4, phantom, "{\scriptstyle x_{202}}" description]
        \arrow[from=u1, to=u2, phantom, "{\scriptstyle x_{022}}" description]
        \arrow[from=u3, to=u4, phantom, "{\scriptstyle x_{122}}" description]
        \end{tikzcd}
    \end{minipage}
\end{center}
In the last diagram, the filler for the biggest square is $x_{221}$ and the filler for the cube is $x_{222}$.

\paragraph{Degenerating a line}

Using our newly developed geometric intuition, we explain how two
substitutions of type 
$\Sub\,\Gamma^\P\,(\Gamma^\P)^\P$, namely
$\R_{\Gamma^\P}$ and $(\R_\Gamma)^\P$, differ. They correspond to the
two different ways of turning a line into a square: given
$(\diamond,x:A)^\P$, $\R_{(\diamond,x:A)^\P}$ produces the square on
the left, $(\R_{\diamond,x:A})^\P$ produces the square on the right. \\
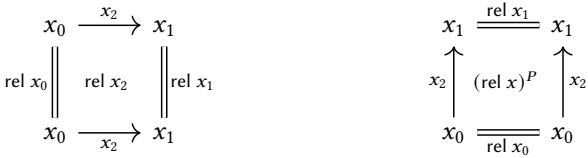
\begin{center}
    \begin{minipage}{0.5\textwidth}
        \begin{center}
            \begin{tikzcd}[row sep=large]
            x_0 \arrow[r, "x_2"] & x_1\\
            x_0 \arrow[u, "{\scriptstyle \rel\;x_0}", equal] \arrow[r, "x_2"'] \arrow[ur, phantom, "{\scriptstyle \rel\;x_2}" description] & x_1 \arrow[u, "\rel\;x_1"', equal]
            \end{tikzcd}
        \end{center}   
    \end{minipage} \hfill
    \begin{minipage}{0.45\textwidth}
        \begin{tikzcd}[row sep=large]
        x_1 \arrow[r, "\rel\;x_1", equal] & x_1\\
        x_0 \arrow[u, "x_2"] \arrow[r, "\rel\;x_0"', equal] \arrow[ur, phantom, "{\scriptstyle (\rel\;x)^P}" description] & x_0 \arrow[u, "x_2"']
        \end{tikzcd}   
    \end{minipage} \hfill
\end{center}
Assuming $\R_{\Gamma^\P} = (\R_\Gamma)^\P$ is incompatible with
injectivity of $\Pi$ (an analogous observation was made by
\cite[p.\ 138]{DBLP:conf/lics/BernardyM12}). To explain this, we need to
know how $\blank^\P$ acts on $\Pi$ types:
\begin{alignat*}{10}
  & (\Pi(x:A).B)^\P[\gamma_P,f_0,f_1] := \\
  & \hspace{2em}\Pi(x_0:A[\o_\Gamma\circ\gamma_P], x_1:A[\i_\Gamma\circ\gamma_P], x_2:A^\P[\gamma_P,x_0,x_1]).B^\P[\gamma_P,x_0,x_1,x_2,f_0\oldapp x_0, f_1\oldapp x_1]
\end{alignat*}
Now we can see that ${\U^\P}^\P[a_{00}, \dots, a_{21}]$ is the type
of two-dimensional relations, parameterised by four codes
$a_{00}$, $a_{01}$, $a_{10}$, $a_{11}$ in $\U$ and four relations
$a_{02}$, $a_{12}$, $a_{20}$, $a_{21}$ between them. We compute as follows.
\begin{alignat*}{20}
  & \rlap{${\U^\P}^\P[a_{00}, \dots, a_{21}] =$} \\
  & \rlap{${\U^\P[a_0,a_1]}^\P[a_{00}, \dots, a_{21}] =$} \\
  & \rlap{$(\El\,a_0\Ra\El\,a_1\Ra\U)^\P[a_{00}, \dots, a_{21}] =$} \\
  & \Pi\Big(&& x_{00} && {}:{} && \El\,a_{00},{} && x_{01} && {}:{} && \El\,a_{01},{} && x_{02} && {}:{} && \El\,(a_{02}\oldapp x_{00}\oldapp x_{01}), \\
  & && x_{10} && {}:{} && \El\,a_{10},{} && x_{11} && {}:{} && \El\,a_{11},{} && x_{12} && {}:{} && \El\,(a_{12}\oldapp x_{10}\oldapp x_{11})\Big).\U^\P[a_{20}\oldapp x_{00}\oldapp x_{10}, a_{21}\oldapp x_{01}\oldapp x_{11}] = \\
  & \Pi\Big(&& x_{00} && {}:{} && \El\,a_{00},{} && x_{01} && {}:{} && \El\,a_{01},{} && x_{02} && {}:{} && \El\,(a_{02}\oldapp x_{00}\oldapp x_{01}), \\
  & && x_{10} && {}:{} && \El\,a_{10},{} && x_{11} && {}:{} && \El\,a_{11},{} && x_{12} && {}:{} && \El\,(a_{12}\oldapp x_{10}\oldapp x_{11})\Big).\El\,(a_{20}\oldapp x_{00}\oldapp x_{10})\Ra\El\,(a_{21}\oldapp x_{01}\oldapp x_{11})\Ra\U
\end{alignat*}
Assuming $\R_{(\diamond,a:\U)^\P} = (\R_{\diamond,a:\U})^\P$, we also
have ${\U^\P}^\P[\R_{(\diamond,a:\U)^\P}] =
{\U^\P}^\P[(\R_{\diamond,a:\U})^\P]$, but the first one is a type of
the form $\Pi(x_{00}:\El\,a_{0}, x_{01} :\El\,a_{0}\dots).\dots$, the
second one is a type of the form $\Pi(x_{00}:\El\,a_{0}, x_{01}
:\El\,a_{1}\dots).\dots$. If $\Pi$ has injectivity (which follows from
normalisation), then for any $a_0$, $a_1$ in $\U$, $\El\,a_0 =
\El\,a_1$.

\paragraph{Symmetry and emergent geometry}

In a syntax for internal parametricity, we either need to postulate
the existence of an infinite hierarchy of $\rel$s from which the
substitutions $\R_\Gamma$, $(\R_\Gamma)^\P$, ${(\R_\Gamma)^\P}^\P$,
$\dots$ can be obtained, or we need to provide another way to relate
$\R_{\Gamma^\P}$ and $(\R_\Gamma)^\P$. We choose the latter:%
\footnote{In fact, it is not clear whether the former is even
possible.  The naive presheaf model of such a theory does not satisfy
the needed computation rule for $\blank^\P$ on $\Pi$, and the syntax
has stuck terms that it is not clear how to compute with.}  we
introduce a new substitution $\S_\Gamma$ (called symmetry) which
satisfies $\S_\Gamma\circ\R_{\Gamma^\P} =
(\R_\Gamma)^\P$. Intuitively, symmetry maps $x_{ij}$ to $x_{ji}$.
\begin{center}
\begin{tikzpicture}
\node (a00)  at (0,0) {$x_{00}$};
\node (a01)  at (0,1) {$x_{01}$};
\node (a10)  at (1.5,0) {$x_{10}$};
\node (a11)  at (1.5,1) {$x_{11}$};
\node (a22)  at (0.75,0.5) {$x_{22}$};
\draw[->] (a00) edge node[left] {$x_{02}$} (a01);
\draw[->] (a00) edge node[below] {$x_{20}$} (a10);
\draw[->] (a10) edge node[right] {$x_{12}$} (a11);
\draw[->] (a01) edge node[above] {$x_{21}$} (a11);
\node (b1) at (2.5,0.5) {};
\node (b2) at (3.5,0.5) {};
\draw[|->] (b1) edge node[above] {$\S$} (b2);
\begin{scope}[shift={(4.5,0)}]
\node (a00)  at (0,0) {$x_{00}$};
\node (a01)  at (0,1) {$x_{10}$};
\node (a10)  at (1.5,0) {$x_{01}$};
\node (a11)  at (1.5,1) {$x_{11}$};
\node (a22)  at (0.75,0.5) {$\mathsf{sym}\,x_{22}$};
\draw[->] (a00) edge node[left] {$x_{20}$} (a01);
\draw[->] (a00) edge node[below] {$x_{02}$} (a10);
\draw[->] (a10) edge node[right] {$x_{21}$} (a11);
\draw[->] (a01) edge node[above] {$x_{12}$} (a11);
\end{scope}
\end{tikzpicture}
\end{center}
It turns out that this operation is enough, and there is no need to
introduce higher dimensional versions of $\S_\Gamma$ (as in
\cite{DBLP:conf/lics/BernardyM12}) or an extra sort of intervals (as
in \cite{DBLP:journals/entcs/BernardyCM15,DBLP:phd/us/Cavallo21}). In
this paper we define a theory with internal parametricity which does
not have explicit geometry in the syntax. Compared to previous
theories with internal parametricity, geometry is \emph{emergent}
rather than explicitly built-in. We do have ways to talk about higher
dimensional cubes (as we saw when iterating $\blank^\P$ on contexts)
but this is nothing special: Martin-Löf type theory also has all
higher dimensional cubes simply because the identity type can be
iterated. E.g.\ the type
$\Id_{(\Id_A\,a_{01}\,a_{11})}\,\big(\mathsf{transport}_{(\Id_A\,a_{01}\,\blank)}\,a_{12}\,(\mathsf{transport}_{(\Id_A\,\blank\,a_{10})}\,a_{02}\,a_{20})\big)\,a_{21}$
expresses the type of fillers of the following two-dimensional square.
\begin{center}
\begin{tikzpicture}
\node (a00)  at (0,0) {$a_{00}$};
\node (a01)  at (0,1) {$a_{01}$};
\node (a10)  at (1,0) {$a_{10}$};
\node (a11)  at (1,1) {$a_{11}$};
\draw[->] (a00) edge node[left] {$a_{02}$} (a01);
\draw[->] (a00) edge node[below] {$a_{20}$} (a10);
\draw[->] (a10) edge node[right] {$a_{12}$} (a11);
\draw[->] (a01) edge node[above] {$a_{21}$} (a11);
\end{tikzpicture}
\end{center}
We have one two-dimensional operation ($\S_\Gamma$) and one equation
about $\S_\Gamma$ that involes three dimensional cubes, but we never
mention anything higher than that.

\paragraph{Obtaining a theory from a model}

Even if higher cubes are not explicitly built into our syntax, our
type theory is informed by an analysis of the cubical set model built on
Bezem-Coquand-Huber (BCH) cubes \cite{DBLP:conf/types/BezemCH13}.

The BCH cube category can be presented using a finite number of basic
operators and equations between them. It is given by the free
2-category generated by the diagram on the left in Figure
\ref{fig:cubecat} and five equations relating the 2-cells.
\begin{figure}
\begin{tikzcd}[row sep=7 em]
    * \\
    *\arrow[u, bend left=100, "\suc"{name=suc}] \arrow[u, bend right=5, "\id"{name=id, right}] \arrow[u, bend right=100, "\;"{name=S, right}, "\suc\circ \suc"' near start]
    \arrow[from=S, to=S, loop right, distance=3 em, start anchor={[yshift=2ex]west}, end anchor={[yshift=-2ex]west}, shorten=1 mm, "\S", Rightarrow]
    \arrow[from=suc, to=id, shift left=4 ex, shorten=2 mm, "\R", Rightarrow]
    \arrow[from=id, to=suc, shift left=5 ex, shorten=2 mm, "\o", Rightarrow]
    \arrow[from=id, to=suc, shift left=3 ex, shorten=2 mm, "\i"', Rightarrow]
\end{tikzcd}\hspace{5em}
\begin{tikzcd}[row sep=7 em]
    \PSh(\square)
    \arrow[d, bend right=100, "\foralle"{name=suc,left}]
    \arrow[d, bend left=5, "\mathsf{id}"{name=id, right}]
    \arrow[d, bend left=100, "\;"{name=S, right}, "\foralle\circ\foralle" near end] \\
    \PSh(\square)
    \arrow[from=S, to=S, loop right, distance=3 em, start anchor={[yshift=2ex]west}, end anchor={[yshift=-2ex]west}, shorten=1 mm, "\Se", Leftarrow]
    \arrow[from=id, to=suc, shift right=4 ex, shorten=2 mm, "\Re"', Rightarrow]
    \arrow[from=suc, to=id, shift right=5 ex, shorten=2 mm, "\e{\o}"', Rightarrow]
    \arrow[from=suc, to=id, shift right=3 ex, shorten=2 mm, "\e{\i}", Rightarrow]
\end{tikzcd}
\caption{The generating 1-cells and 2-cells of the BCH cube category
  $\square$ as a 2-category (left). Structure on presheaves over this
  category (right).
  The two occurrences of $*$ represent the same object, depicted twice solely for presentation, and the same holds for $\PSh(\square)$.}
\label{fig:cubecat}
\end{figure}
This 2-category has one 0-cell $*$, and can be seen as a 1-category
where objects are given by 1-cells from $*$ to itself, and morphisms
are given by the 2-cells. In this presentation we have numbered
dimensions instead of named dimensions (see
\cite{DBLP:conf/RelMiCS/BuchholtzM17} for a comparison of different
presentations).  The objects of the cube category are natural numbers
given by $0 = \mathsf{id}$, $1 = \mathsf{suc}\circ\mathsf{id}$, $2 =
\mathsf{suc}\circ\mathsf{suc}\circ\mathsf{id}$, and so
on. Degeneracies are generated by $\R$, e.g.\ there is one map $\R$
from $1$ to $0$, there are two maps from $2$ to $1$, namely
$\id_{\suc}\bullet\R$ and $\R\bullet\id_{\suc}$ (where
$\blank\bullet\blank$ denotes horizontal composition). Face maps are
similarly generated by $\o$ and $\i$, and there is a symmetry
operation $\S$.

The category of presheaves over the BCH cube category supports the exact same structure with
maps in the other direction (diagram on the right in Figure
\ref{fig:cubecat}). Here $\foralle$ is precomposition by
$\mathsf{suc}$, and the natural transformations are named after their
generating 2-cells. This picture can be reified into new operations of
the type theory: we add a strict morphism from the syntax to itself
corresponding to $\foralle$, natural transformations $\Re$, $\oe$,
$\ie$ and $\Se$, and the five extra equations. The new equations
and the other operators that we add are those which are justified by this
presheaf model: the morphism $\foralle$
strictly respects the substitution calculus (the category with families, CwF), $\top$, $\Sigma$, strict
identity $\Eq$, $\Bool$ and $\K$. Our model does
not justify an equation such as $\foralle(\Pi\,A\,B) = \Sigma\big(\Pi\,(\foralle A)\,(\foralle B)\big)\dots$,
but only the analogous isomorphism, so for $\Pi$ we add new operators and equations expressing this isomorphism.
Similarly, $\foralle\U$ is described by a section-retraction pair.
We call the collection of these new operations the \emph{global theory}, as it
involves operations on contexts and substitutions.
We directly obtained the global theory from the presheaf model over BCH cubes,
thus it is immediately justified by this model.

\paragraph{Local theory}

Multi-modal type theory \cite{DBLP:journals/lmcs/GratzerKNB21} gives a
generic way to construct a type theory from a CwF morphism such as $\suc$. It uses
that precomposition $\foralle = \suc^*$ has a left adjoint
$\suc_{!}$ (the left Kan extension), and this has a dependent right
adjoint. However this theory is still non-local as it introduces extra
(un)lock operations on contexts. A presentation of internal
parametricity using this method is \cite{DBLP:phd/us/Cavallo21}, where the
lock operation of multi-modal type theory becomes context extension with
an interval variable.

In our case we can define a version of our theory with only
local operations, that is, operations that do not change the
context. This is what we call the \emph{local theory}. It is specified
by the exact same data as the global theory, but now $\forall$ is not
a morphism from the syntax to itself, but a morphism from the standard
model to itself, internal to presheaves over the syntax. We explain
this in detail.

Any category of presheaves has a type-theoretic internal language. For example, when we write $A :
\Set$ in this internal language, externally this means that $A$ is a
presheaf. Similarly, the internal $B : A \ra \Set$ means a dependent
presheaf over $A$ externally. If the base category $\C$ of the presheaf
model is not only a category, but a CwF, then internally we have $\Ty : \Set$ and $\Tm :
\Ty\ra\Set$ which externally are defined by the presheaf of types and the
dependent presheaf of terms in $\C$. This is the main idea of
two-level type theory \cite{DBLP:conf/csl/AltenkirchCK16,DBLP:journals/corr/AnnenkovCK17}.

If $\C$ has $\Sigma$-types, then internally we have $\Sigma : (A :
\Ty)\ra(\Tm\,A\ra\Ty)\ra\Ty$ together with an isomorphism
$(a:\Tm\,A)\times\Tm\,(B\,a) \cong \Tm(\Sigma\,A\,B)$. In this case
(still internally), $\Ty$ and $\Tm$ form a universe closed under $\Sigma$
types. \cite{DBLP:conf/fscd/BocquetKS23} call this a higher-order
model of type theory with $\Sigma$ types. The situation is analogous
for other type formers, e.g., if $\C$ has $\Pi$-types, then internally we
have $\Pi : (A : \Ty)\ra(\Tm\,A\ra\Ty)\ra\Ty$ together with an
isomorphism $((a:\Tm\,A)\ra\Tm\,(B\,a))\cong\Tm\,(\Pi\,A\,B)$.

Now we define a CwF internal to presheaves over $\C$. We call this the
internal standard model. Contexts in this model are given by $\Ty$, a
type in a context $\Omega$ is a function $\Tm\,\Omega\ra\Ty$, a term
in context $\Omega$ of type $A$ is a dependent function
$(\omega:\Tm\,\Omega)\ra\Tm(A\,\omega)$. Context extension is given by
$\Sigma$, hence we need that $\Ty$ is closed under $\Sigma$, which in
turn needs that $\C$ has $\Sigma$ types. This standard model is a
generalisation of the set model (or type model, or metacircular model)
\cite{DBLP:conf/popl/AltenkirchK16} and is a variant of the telescopic
contextualisation of \cite{DBLP:conf/fscd/BocquetKS23}.

The $\forall$ of the local theory is specified by a CwF-morphism from
this standard model to itself. Internally written, it maps contexts to
contexts ($\forall:\Ty\ra\Ty$), types to types ($\foralld :
(\Tm\,A\ra\Ty)\ra\Tm\,(\forall A)\ra\Ty$), and so on. We add a ``d''
suffix to the operation on types to distinguish from the one on
contexts. Externally, these operations are natural tranformations
described as follows.
\[
\infer{\forall A : \Ty\,\Gamma}{A : \Ty\,\Gamma} \hspace{2em} \infer{(\forall A)[\sigma] = \forall(A[\sigma])}{} \hspace{2em} \infer{\foralld B : \Ty\,(\Gamma,x:\forall A)}{B:\Ty\,(\Gamma,x:A)} \hspace{2em} \infer{(\foralld B)[\sigma] = \foralld(B[\sigma\uparrow])}{}
\]
We obtain all of the local theory this way: we start with a strictly
democratic CwF $\C$ with $\top$, $\Sigma$, $\Eq$, $\Pi$, $\U$ and
$\Bool$ (we call this the core theory). Internally to presheaves over
$\C$, we have the standard model of this core theory. Now $\forall$,
$\foralld$ and the other new operations and equations providing
internal parametricity are specified by a core theory morphism from
this standard model to itself. Just as in the case of the global
theory, this morphism respects the CwF structure, $\top$, $\Sigma$,
$\Eq$, $\K$ and $\Bool$ strictly, $\Pi$ up to an isomorphism, and $\U$ up to
section-retraction. Note that this only gives a specification of the
local theory, and does not directly provide a model of it. We justify
the local theory by deriving its syntax from the syntax of the global
theory which we localise using $\Re$.

The local theory is truly local: it does not mention contexts and it can
be described as a second-order generalised algebraic theory (SOGAT)
\cite{DBLP:journals/corr/abs-1904-04097}. From this SOGAT we obtain a
first-order GAT in a way that makes sure that all operations are
stable under substitution. As far as we know, our local theory is the
first non-substructural type theory describing presheaves over BCH
cubes. We distinguish the corresponding operations of the local and
global theories by writing those of the global theory in
{\color{brickred} brick red colour}.

\paragraph{Span-based parametricity}

In our global theory, $\foralle\Gamma$ exactly corresponds to the
$\Gamma^\P$ of the external parametricity translation. For types
however we don't compute parametricity relations, but parametricity spans:
$\foralle A : \Ty\,(\foralle \Gamma)$ together with maps $\ke_A :
\Tm\,(\foralle\Gamma,\foralle A)\,(A[\ke_\Gamma])$. Similarly, in the
local theory, $\forall A$ is a type with the structure of a span $A
\overset{\o_A}{\longleftarrow} \forall
A\overset{\i_A}{\longrightarrow} A$. We can recover the relational
version (Bridge type) using the strict identity
type $\Eq$:
\[
A^\P\,a_0\,a_1 := \Sigma(a:\forall A).\Eq_A\,(\o_A\,a)\,a_0 \times \Eq_A\,(\i_A\,a)\,a_1.
\]
Just as $\blank^\P$, the operation $\forall$ computes definitionally
on $\top$, $\Sigma$, $\Eq$ and $\Bool$. For example, $\forall$ of a
$\Sigma$ is equal to a $\Sigma$ of $\forall$s. On function types we
have the span-preservation variant of usual relation-preservation
$(A\Ra B)^\P\,f_0\,f_1 =
\Pi(x_0:A[\o],x_1:A[\i],x_2:A^\P\,x_0\,x_1).B^\P\,(f_0\oldapp
x_0)\,(f_1\oldapp x_1)$ saying that related inputs are mapped to
related outputs. An element of $\forall(A\Ra B)$ corresponds to a
function $t$ from $\forall A$ to $\forall B$, and functions $t_k$ from
$A[k]$ to $B[k]$ such that the following diagram commutes.
\[\begin{tikzcd}
A[\o] \arrow[d,"t_0"'] & \forall A\arrow[l,"\o_A"']\arrow[d,"t"]\arrow[r,"\i_A"] & A[\i]\arrow[d,"t_1"] \\
B[\o] & \forall B \arrow[l,"\o_B"]\arrow[r,"\i_B"'] & B[\i]
\end{tikzcd}\]
This correspondence holds only up to isomorphism in the model and thus
in our theory.

In the external parametricity translation, we have $\U^\P\,a_0\,a_0 =
(\El\,a_0\Ra\El\,a_1\Ra\U)$. The main reason that we use span-based
instead of relation-based parametricity is that in our model this is
not an equality, only a logical equivalence.\footnote{Although if we
observe that both sides are the type of objects of some category, we
can say that it extends to an equivalence of categories.} We have
maps in both directions, but the composite map
\[
(\El\,a_0\Ra\,\El\,a_1\Ra\U)\hspace{1em}\ra\hspace{1em}\U^\P\,a_0\,a_1\hspace{1em}\ra\hspace{1em}(\El\,a_0\Ra\,\El\,a_1\Ra\U)
\]
is not the identity (and neither is the other roundtrip). However, if
we replace relations by spans, we do have that the analogous composite
map for spans
\[
%\Sigma(a,a_k:\U).(\El\,a\Ra\El\,a_k)\hspace{2.0em}\ra\hspace{2.0em}\forall\U\hspace{2.0em}\ra\hspace{2.0em}\Sigma(a,a_k:\U).(\El\,a\Ra\El\,a_k)
\Sigma(a,a_0,a_1:\U).(\El\,a\hspace{-0.1em}\Ra\hspace{-0.1em}\El\,a_0)\times(\El\,a\hspace{-0.1em}\Ra\hspace{-0.1em}\El\,a_1)\hspace{0.0em}\ra\hspace{0.0em}\forall\U\hspace{0.0em}\ra\hspace{0.0em}\Sigma(a,a_0,a_1:\U).(\El\,a\hspace{-0.1em}\Ra\hspace{-0.1em}\El\,a_0)\times(\El\,a\hspace{-0.1em}\Ra\hspace{-0.1em}\El\,a_1)
\]
is the identity. The intuition for
why this works for spans and not for relations is that presheaves are
span-based by nature. Given a presheaf $\Gamma$ over BCH-cubes
$\square$, we denote the action on objects $\Gamma : \square\ra\Set$ and
the action on morphisms $\blank[\blank]_\Gamma :
\Gamma\,I\ra\square(J,I)\ra\Gamma\,J$. Now, at levels $0$ and $1$ we
have a span
$\Gamma\,0 \overset{\blank[0]_\Gamma}{\longleftarrow} \Gamma\,1 \overset{\blank[1]_\Gamma}{\longrightarrow} \Gamma\,0$
instead of a relation $\Gamma\,0\ra\Gamma\,1\ra\Set$. Relation-based
presheaves are called Reedy fibrant (relative to families of sets as the underlying notion of ``fibration'')
\cite{DBLP:journals/corr/KrausS17}, and it should be possible to construct a
Reedy fibrant presheaf model of internal parametricity, but we leave this for future
work. A model of internal parametricity based on refined presheaves
similar to Reedy fibrant ones is
\cite{DBLP:journals/entcs/BernardyCM15}.

The other roundtrip for the correspondence $\forall\U\leftrightarrow
\text{Span}$ is unfortunately not identity in our model. Thus we
justify this correspondence up to a section-retraction pair.

\paragraph{Metatheory}

As our global theory arose from a presheaf model, it is not surprising
that it is modelled by the exact same presheaf category. We extend
gluing \cite{DBLP:conf/rta/KaposiHS19} to the global theory, and
define a global section functor satisfying the necessary conditions
from the syntax to our presheaf model. We know that the syntaxes of
the global and local theories are isomorphic, hence, as a result, our
local theory satisfies canonicity and has an external parametricity
translation. We conjecture that a version of our theory without
equality reflection satisfies normalisation.

\subsection{Structure of the Paper}

After summarising related work and our notations, we introduce our
local theory in Section \ref{sec:local}, and describe some
applications including well-known usages of internal
parametricity. This section can be understood without prior
familiarity with models of type theory or presheaves. For the rest of
the paper we try to be as self-contained as possible, and refer to the
relevant literature.

In Section \ref{sec:global}, we define the global version of the
theory as a generalised algebraic theory (GAT). The syntax of the
theory is given by the initial algebra (model) which exists for any
GAT. We also show that the
global theory has a presheaf model. Section \ref{sec:iso} shows the
isomorphism of the local and global syntaxes. Then in Section \ref{sec:gluing} we prove that our
global theory has a gluing model and as a consequence satisfies
canonicity: every closed term of the boolean type is equal to either
true or false. By the previous isomorphism this result holds for both the local and global
theories.

\subsection{Related Work}

The first type theory with internal parametricity was defined by
Bernardy and Moulin \cite{DBLP:conf/lics/BernardyM12}. It contains a
syntax for arbitrary dimensional cubes. The $\apd$ operator in our
local syntax is very similar to their double bracket operator, but
our theory only mentions cubes up to dimension three. Bernardy and
Moulin simplified their syntax later using named dimensions
\cite{DBLP:conf/icfp/BernardyM13} and further refined it using a sort
of intervals \cite{DBLP:journals/entcs/BernardyCM15} similar to that
of cubical type theories
(e.g.\ \cite{DBLP:conf/types/CohenCHM15}). Using the same (BCH)
cube category as the first cubical set model of univalence
\cite{DBLP:conf/types/BezemCH13}, the paper
\cite{DBLP:journals/entcs/BernardyCM15} defines a presheaf model of
internal parametricity. It uses a refined notion of dependent presheaf
for interpreting types similar to Reedy fibrancy
\cite{DBLP:journals/corr/KrausS17}. Our presheaf model uses ordinary
presheaves and avoids the need of Reedy fibrancy by having span-based
instead of relational
parametricity. \cite{DBLP:journals/entcs/BernardyCM15} has an equation
``SURJ-TYP'' the analog of which we were not able to justify in our
model. This would correspond to $\unspan$ being an isomorphism, not
only a section (see (7) in Problem \ref{con:psh_global}). It
seems that having Reedy fibrant types does not make a difference in this respect.

Cavallo and Harper
\cite{DBLP:journals/lmcs/CavalloH21,DBLP:phd/us/Cavallo21} define a
cubical type theory with univalence and internal parametricity at the
same time. They support a double-presheaf model with cartesian cubes
for the identity type and BCH cubes for parametricity. They justify
relational parametricity, but the correspondence between the logical
relation at $\U$ and relation space only holds up to internal
equivalence (and hence propositional equality, by univalence), and not
definitional section-retraction as in our theory. This is enough to
derive the consequences of internal parametricity, however. Van
Muylder, Nuyts and Devriese \cite{new} extend Cubical Agda with
internal parametricity following Cavallo and Harper. Inside this
theory they shallowly embed a ``relational observational type theory''
in which logical relations are computed as in
\cite{DBLP:conf/lics/BernardyM12}, but it does not feature iterated
parametricity.

Nuyts et
al.\ \cite{DBLP:journals/pacmpl/NuytsVD17,DBLP:conf/lics/NuytsD18}
analyse the presheaf model of internal parametricity, and define type
theories where the parametricity relation and the (non-univalent)
identity type are special cases of a general construction
``relatedness''. These syntaxes use two different kind of $\Pi$ types
(parametric and non-parametric ones) and there is no proof of
canonicity.
% Nuyts et
% al.\ \cite{DBLP:journals/pacmpl/NuytsVD17,DBLP:conf/lics/NuytsD18}
% define type theories with a different notion of internal parametricity
% than that of \cite{DBLP:journals/entcs/BernardyCM15}. These theories
% satisfy identity extension which implies that (using some
% impredicativity), Church-encoded booleans actually have an induction
% principle. In the weaker setting, we only have that Church-encoded
% booleans are isomorphic to inductive booleans (if the latter are part
% of the theory). Our paper also provides internal parametricity in this
% weaker sense.

Our global syntax is very close to the naive syntax in
\cite{DBLP:conf/types/AltenkirchK15}. By closely following a model, we
make sure that we do not miss any equations, and we manage to prove
canonicity, even being ``naive'' in their sense.

Recent work on cubical type theories
\cite{DBLP:conf/types/CohenCHM15,DBLP:journals/mscs/AngiuliBCHHL21,DBLP:journals/jfp/VezzosiMA21}
has used different cube categories that contain diagonals and
sometimes connections as well, due to their advantages when
formulating higher inductive types.  The resulting presheaf categories
satisfy a different computation rule for $\forall$ of $\Pi$ (function
extensionality), which is correct homotopically but inappropriate for
parametricity.  On the other hand, earlier work on cubical homotopy
theory used a cube category lacking symmetries, which also fails to
have our desired computation rule for $\forall$ of $\Pi$.  The BCH
cube category is ``just right''.

Finally, although we explicitly discuss only ``binary'' parametricity,
one can consider $n$-ary parametricity for any natural number $n$, in
which case there are $n$ possible values for $k$ wherever it appears.
Unary parametricity is also common in the literature.  Nothing that we
say should be sensitive to the choice of $n$, and often even our
notation can be applied directly in the $n$-ary case.

\subsection{Metalanguage and Notation}
\label{sec:metalang}

Our metalanguage is extensional type theory with quotients and
propositional extensionality (unlike in the above paragraph ``Local
theory'', this extensional type theory is not the outer level of a
two-level type theory). Our constructions can be also understood as
taking place in a constructive set theory. We use Agda-style notation
with implicit arguments usually omitted or written in curly braces
$\{{\dots}\}$ and we employ implicit coercions and overloaded
projections.

We use categories with families (CwFs,
\cite{DBLP:journals/corr/abs-1904-00827}) as the notion of model of
type theory. The components of the category part are denoted $\Con$,
$\Sub$, $\blank\circ\blank$, $\id$, the terminal object (empty
context) $\diamond$, the empty substitution $\epsilon$. The families
of types and terms are $\Ty$, $\Tm$, their instantiation of
substitution operations are both denoted $\blank[\blank]$. We write
$\Gamma\ext A$ for the context $\Gamma$ extended by the type $A$. We
write $(\sigma,t) : \Sub\,\Delta\,(\Gamma\ext A)$ for
$\sigma:\Sub\,\Delta\,\Gamma$ and $t : \Tm\,\Delta\,(A[\sigma])$, and
denote the projections by $\p : \Sub\,(\Gamma\ext A)\,\Gamma$ and
$\q:\Tm\,(\Gamma\ext A)\,(A[\p])$. We write $(\sigma\uparrow)$ for
$(\sigma\circ\p,\q)$. We write $\p^2 = \p\circ\p$, $\p^3 = \p^2\circ\p$,
and natural numbers for De Bruijn indices: $n = \q[\p^n]$. We denote
an isomorphism between two contexts by $\sigma:\Gamma\cong\Delta$
which means that $\sigma:\Sub\,\Gamma\,\Delta$ and there is also a
$\sigma^{-1}:\Sub\,\Delta\,\Gamma$ and both compositions
$\sigma\circ\sigma^{-1}$ and $\sigma^{-1}\circ\sigma$ are the
identity. Similarly, for types we write $A\cong B$ to mean that we
have terms $t:\Tm\,(\Gamma\ext A)\,(B[\p])$ and $t^{-1}:\Tm\,(\Gamma\ext
B)\,(A[\p])$ such that $t[\p,t^{-1}] = \q$ and $t^{-1}[\p,t] = \q$.

\section{The local theory and applications}
\label{sec:local}

In this section we list and explain the rules of our local theory with
internal parametricity and show how to apply it to derive consequences
of parametricity. This section can be understood without previous
knowledge of models of type theory.

\begin{figure}
  \begin{gather*}
    \infer{\top}{} \hspace{2em}
    \infer{\tt:\top}{} \hspace{2em}
    \infer{t = \tt}{t : \top} \\
   \infer{\Sigma(x:A).B}{A && x:A\vdash B} \hspace{2em}
    \infer{(u,v) : \Sigma(x:A).B}{u : A && v : B[v\mapsto u]} \hspace{2em}
    \infer{\pi_1\,t:A}{t : \Sigma(x:A).B} \hspace{2em}
    \infer{\pi_2\,t:B[x\mapsto \pi_1\,t]}{t : \Sigma(x:A).B} \\
    \pi_1\,(u,v) = u\hspace{2em}
    \pi_2\,(u,v) = v\hspace{2em}
    t = (\pi_1\,t,\pi_2\,t) \\
    \infer{\Eq_A\,a_0\,a_1}{a_0 : A && a_1 : A} \hspace{2em}
    \infer{\refl_a : \Eq_A\,a\,a}{a : A} \hspace{2em}
    \infer{a_0 = a_1}{e : \Eq_A\,a_0\,a_1} \hspace{2em}
    \infer{e = \refl_a}{e:\Eq_A\,a\,a} \\
    \infer{\Pi(x:A).B}{A && x:A\vdash B} \hspace{2em}
     \infer{\lambda x.t:\Pi(x:A).B}{x:A\vdash t:B} \hspace{2em}
     \infer{t\oldapp u:B[x\mapsto u]}{t:\Pi(x:A).B && u:A} \\
    (\lambda x.t)\oldapp u = t[x\mapsto u] \hspace{2em} 
    t = \lambda x.t\oldapp x \\
    \infer{\U}{} \hspace{2em}
    \infer{\El\,a}{a :\U} \hspace{2em}
    \infer{\c\,A : \U}{A} \hspace{2em}
    \El\,(\c\,A) = A \hspace{2em}
    \c\,(\El\,a) = a \\
    \infer{\Bool}{} \hspace{1.2em}
    \infer{\true:\Bool}{} \hspace{1.2em}
    \infer{\false:\Bool}{} \hspace{1.2em} 
    \infer{\ite_{x.A}\,t\,u\,v : A[x\mapsto t]}{x:\Bool\vdash A \hspace{1.2em} t:\Bool \hspace{1.2em} u : A[x\mapsto \true] \hspace{1.2em} v : A[x\mapsto\false]} \\
    \ite\,\true\,u\,v = u \hspace{2em}
    \ite\,\false\,u\,v = v \hspace{2em}
  \end{gather*}
  \caption{The core theory. See Definition \ref{def:coreExt} for an external description.}
  \label{fig:core}
\end{figure}

\begin{figure}
  \begin{alignat*}{10}
    & \rlap{The new operations:} \\
    & \rlap{$\infer{\forall A}{A} \hspace{1.5em}
    \infer{\ap(x.t)\,a_2 : \forall B}{B \hspace{1.5em} x:A\vdash t:B \hspace{1.5em} a_2 :\forall A} \hspace{1.5em}
    \infer{\foralld(x.B)\,a_2}{x:A\vdash B \hspace{1.5em} a_2 : \forall A} \hspace{1.5em}
    \infer{\apd(x.t)\,a_2 : \foralld(x.B)\,a_2}{x:A\vdash B \hspace{1.5em} x:A\vdash t:B \hspace{1.5em} a_2 :\forall A} \hspace{2em}$} \\
    & \rlap{$\infer{k_A\,a_2 : A}{a_2:\forall A} \hspace{2em}
    \infer{\R_A\,a : \forall A}{a:A} \hspace{2em} 
    \infer{\S_A\,a_{22} : \forall (\forall A)}{a_{22}:\forall (\forall A)} \hspace{2em}
    \infer{\unspan\,A_k\,A\,x.t_k:\forall\U}{A_k && A && x:A\vdash t_k : A_k} \hspace{2em}$} \\
    & \rlap{$\infer{\mkpi\,a_2\,t_k\,t : \foralld(x.\Pi(y:B).C)\,a_2}
    {\arraycolsep=1.4pt\begin{array}{l}
        x:A\vdash B \\
        x:A,y:B\vdash C \\
        a_2:\forall A \\
        t_k : \Pi(y:B[x\mapsto k_A\,a_2]).C[x\mapsto k_A\,a_2] \\
        t : \Pi(y_2:\foralld(x.B)\,a_2).\foralld((x,y).C)\,(a_2,y_2) \\
        y_2:\foralld(x.B)\,a_2 \vdash t_k\oldapp\d{k}_{x.B}\,a_2\,y_2 = \d{k}_{(x,y).C}\,(a_2,y_2)\,(t\oldapp y_2)
    \end{array}} \hspace{1.7em}
    \infer
        {\arraycolsep=1.4pt\begin{array}{l}
            \d{k}_{x.B}\,a_2\,b_2 : B[x\mapsto k_A\,a_2] \\
            \d{k}_{x.B}\,a_2\,b_2 := \pi_2\,(k_{\Sigma(x:A).B}\,(a_2,b_2)) \\
            \hspace{1em}
        \end{array}}
        {\arraycolsep=1.4pt\begin{array}{l}
          \text{An abbreviation:\hspace{7em}} \\[0.5em]
            x:A\vdash B \\
            a_2:\forall A \\
            b_2:\foralld(x.B)\,a_2
        \end{array}}$} \\
    & \rlap{The new equations:} \\
    & \text{Core}\hspace{1.5em} &&
    \ap(x.g[y\mapsto f])\,a_2 = \ap(y.g)\,(\ap(x.f)\,a_2) \hspace{2em} 
    \ap(x.x)\,a_2 = a_2 \hspace{2em}
    \forall \top = \top \hspace{2em} \\
    & && \foralld(x.C[y\mapsto f])\,a_2 = \foralld(y.C)\,(\ap(x.f)\,a_2) \hspace{1.5em}
    \apd(x.t[y\mapsto f])\,a_2 = \apd(y.t)\,(\ap(x.f)\,a_2) \hspace{2em} \\
    & && \forall (\Sigma(x:A).B) = \Sigma(x_2:\forall A).\foralld(x.B)\,x_2 \hspace{2em}
    \ap(x.(u,v))\,a_2 = (\ap(x.u)\,a_2,\apd(x.v)\,a_2) \hspace{2em} \\
    & \text{Const} && \foralld(\_.B)\,a_2 = \forall B\hspace{2em} 
    \infer{\apd(x.t)\,a_2 = \ap(x.t)\,a_2}{B && x:A\vdash t:B}\hspace{2em} \\
    & k,\R,\S && k_B\,(\ap(x.f)\,a_2) = f[x\mapsto k_A\,a_2] \hspace{2em}
    \ap(x.f)\,(\R_A\,a) = \R_B\,(f[x\mapsto a]) \hspace{2em} \\
    & && \ap\big(x_2.\ap(x.f)\,x_2\big)\,(\S_A\,a_{22}) = \S_B\,\big(\ap(x_2.\ap(x.f)\,x_2)\,a_{22}\big) \hspace{2em}
    k_A\,(\R_A\,a) = a \hspace{2em} \\
    & && k_{\forall A}\,(\S_A\,a_{22}) = \ap(x_2.k_A\,x_2)\,a_{22} \hspace{2em}
    \S_A\,(\R_{\forall A}\,a_2) = \ap(x.\R_A\,x)\,a_2 \hspace{2em}
    \S_A\,(\S_A\,a_{22}) = a_{22} \hspace{2em} \\
    & && \S_{\forall A}\,\big(\ap(x_{22}.\S_A\,x_{22})\,(\S_{\forall A}\,a_{222})\big) = \ap(x_{22}.\S_A\,x_{22})\,\big(\S_{\forall A}\,(\ap(x_{22}.\S_A\,x_{22})\,a_{222})\big) \hspace{2em} \\[0.3em]
    & \Sigma && \foralld(x.\Sigma(y:B).C)\,a_2 = \Sigma(y_2:\foralld(x.B)\,a_2).\foralld((x,y).C)\,(a_2,y_2) \\
    & && \apd(x.(u,v))\,a_2 = (\apd(x.u)\,a_2,\apd(x.v)\,a_2) \hspace{2em} \\[0.3em]
    & \Eq && \foralld(x.\Eq_B\,u\,v)\,a_2 = \Eq_{\foralld(x.B)\,a_2}\,(\apd(x.u)\,a_2)\,(\apd(x.v)\,a_2) \hspace{2em} \\[0.3em]
    & \Pi && \d{k}_{x.\Pi(y:B).C}\,a_2\,(\mkpi\,a_2\,t_k\,t) = t_k  \hspace{2em} \\
    & && \lambda y_2.\apd\big((x,f,y).f\oldapp y\big)\,(a_2,\mkpi\,a_2\,t_k\,t,y_2) = t \\
    && & \mkpi\,a_2\,\big(\d{k}_{x.\Pi(y:B).C}\,a_2\,t_2\big)\,\big(\lambda y_2.\apd((x,f,y).f\oldapp y)\,(a_2,t_2,y_2)\big) = t_2 \\[0.3em]
    & \U && \El\,(k_\U\,(\unspan\,A_k\,A\,t_k)) = A_k \hspace{2em}
    \foralld(x.\El\,x)\,(\unspan\,A_k\,A\,t_k) = A \hspace{2em} \\
    & && \d{k}_{x.\El\,x}\,(\unspan\,A_k\,A\,t_k)\,a = t_k[x\mapsto a] \\[0.3em]
    & \Bool && \forall\Bool = \Bool \hspace{2em}
    \ap(x.\true)\,a_2 = \true \hspace{2em}
    \ap(x.\false)\,a_2 = \false \hspace{2em} \\
    & && \apd(x.\ite_{y.B}\,t\,u\,v)\,a_2 = \ite_{y.\foralld(x.B)\,a_2}\,(\ap(x.t)\,a_2)\,(\ap(x.u)\,a_2)\,(\ap(x.v)\,a_2) \hspace{2em}
  \end{alignat*}
\caption{The rules extending the core theory with internal
  parametricity (local theory). The equations are grouped by type former.}
\label{fig:local}
\end{figure}

We list the operations and equations of our core type theory in Figure \ref{fig:core} and the rules
providing internal parametricity in Figure \ref{fig:local}. The
notation can be understood as listing the operations and equations of
a second-order generalised algebraic theory (SOGAT)
\cite{DBLP:journals/corr/abs-1904-04097,DBLP:conf/fscd/BocquetKS23}. This
is why we don't have to list rules about contexts or
substitutions. (The external presentation of the core theory is Definition \ref{def:coreExt},
the external presentation of the local theory is described in Section \ref{sec:usingglobal}.) We have two sorts, types
are denoted $A$, terms are $t : A$. In case an operation is a binder,
some of its arguments have extra dependencies listed before a
$\vdash$. Some operations have implicit arguments, for example the
constructor for $\Sigma$ types $(\blank,\blank)$ has two implicit
inputs $A$ and $B$ which we omit for readability. Similarly, the
equation $\pi_1\,(u,v) = u$ has four implicit arguments: $A, B, u, v$.
We omit the premises of most equations. Note that the $\eta$ rule for
$\Sigma$ types written $t = (\pi_1\,t,\pi_2\,t)$ only makes sense for
$t$ having a $\Sigma$ type, so we don't have to add this as an
explicit assumption. If one views Figures \ref{fig:core} and \ref{fig:local} as
constructors for a syntax, then we have well-typed (intrinsic) terms
which are quotiented by conversion.

% The rules can also be understood as typing rules for a presyntax
% which itself is implicitly defined by these rules. But in this case
% we implicitly assume all the obvious assumptions (e.g. $\Sigma$ has
% three arguments, the context $\Gamma$ which is implicit, then $A$
% and $B$).

\paragraph{Ad Figure \ref{fig:core}}

We have extensional Martin-Löf type theory with unit type $\top$,
$\Sigma$ types, identity types with reflection and uniqueness
of identity (UIP), $\Pi$ types which all satisfy $\beta$ and $\eta$
rules. We have a Coquand-universe \cite{coquandUniverse} which we
don't index for convenience, but everything in this paper can be
redone using universe indexing. In that case we would need to index
types (and terms) by their level as well. Finally, we have a type of
booleans with dependent elimination $\ite$.

\paragraph{Ad Figure \ref{fig:local}}

We have a nondependent $\forall$ on types which computes the logical
span from a type. We know that nondependent maps (terms $x:A\vdash
t:B$ where $B$ does not depend on $x$) preserve logical spans, we call
this $\ap$ as an homage to the ``apply on path'' congruence operation
in HoTT. For types depending on a variable, we have dependent logical
spans $\foralld$ which depend on a base logical span. These are also
preserved by terms witnessed by $\apd$ (this is our main internal parametricity operation). Then we have projections from
$\forall A$ to $A$, in the binary case $k$ can be either $\o$ or
$\i$. The $k_A$ provide the legs of the logical span $\forall A$. (The
dependent version of this projection can be derived using $\Sigma$
types.) We can build degenerate logical spans using $\R$ and we can
apply a symmetry operation $\S$ to a double-span (two-dimensional
span). We explain $\mkpi$ and $\unspan$ below. The $k$, $\R$ and $\S$
operations are natural with respect to $\ap$ and satisfy five
additional equations (which reflect the equations of the BCH cube category):
\begin{enumerate}
\item the bases of a degenerate span are the same as the element we started with,
\item the two different ways of taking the bases of a double-span are related by $\S$,
\item the two different ways of degenerating a span into a double-span are related by $\S$,
\item double symmetry is identity,
\item we can apply symmetry to a triple span in two different ways: we
  can swap dimensions $0$ and $1$, or we can swap dimensions $1$ and $2$; now
  first swapping $0-1$, then $1-2$, then $0-1$ is the same as first
  swapping $1-2$, then $0-1$, then $1-2$. Combined with naturality of
  $\S$, this ensures that the induced symmetries of an $n$-fold span
  form the $n$-ary symmetric group.
\end{enumerate}
In addition to $\forall$, $\foralld$, $\ap$, $\apd$, $k$, $\R$, $\S$,
we have two more operations which concern $\Pi$ types and the
universe ($\mkpi$ and $\unspan$). We will explain them below.

The core equations say that $\ap$, $\foralld$ and $\apd$ are
functorial with respect to nondependent maps, $\forall$ on unit
returns unit and $\forall$, $\ap$ on $\Sigma$ types is pointwise.

The constant equations express that $\forall$ is a special case of
$\foralld$ when the type is actually nondependent. We have an
analogous equation relating $\apd$ and $\ap$.

For $\Sigma$ types, $\foralld$ is defined in a pointwise way. $C$ has
two dependencies $x:A,y:B\vdash C$, so we have $x:A,y:\forall
B\vdash\foralld(y.C)$, which is not what we want. Instead, we collect
$x$ and $y$ into a $\Sigma$ type and use $w:\Sigma(x:A).B\vdash
C[x\mapsto\pi_1\,w,y\mapsto\pi_2\,w]$ which we abbreviate
$(x,y):\Sigma(x:A).B\vdash C$ by ``pattern matching'' on $w$. $\apd$
preserves pairing and preservation of $\pi_1$ and $\pi_2$ are
provable.

Just as $\Sigma$, strict identity $\Eq$ is strictly
preserved. Preservation of $\refl$ is automatic by UIP.

$\Pi$ types are only preserved by $\foralld$ up to
isomorphism. However we don't have that $\forall(\Pi(x:B).C)$ is
isomorphic to $\Pi(y_2:\forall B).\foralld(y.C)\,y_2$ because we can
use $k_{\Pi(y:B).C}$ to obtain a function $\Pi(y:B).C$ from an element
of the first type, but we cannot obtain such a function from an
element of the second type. Hence we have to add more information to
the second type: functions at the bases that are compatible with the
function at the apex. Moreover $B$ itself can be dependent. So the
final statement is that the obvious map from
$\foralld(x.\Pi(y:B).C)\,a_2$ to a $t$ and $t_k$s that make the
following diagram commute is an isomorphism.
\[\begin{tikzcd}
  B[k_A\,a_2] \arrow[r,"t_k"] & C[k_A\,a_2] \\
  (y_2:\foralld(x.B)\,a_2) \arrow[u,"\d{k}_{x.B}\,a_2"]  \arrow[r,"t"] & \foralld((x,y).C)\,(a_2,y_2) \arrow[u,"\d{k}_{(x,y).C}\,(a_2{,}y_2)"']
\end{tikzcd}\]
The $t$ is computed as follows.
\[
\infer{t := \lambda y_2.\apd((x,f,y).f\oldapp y)\,(a_2,t_2,y_2) : \Pi(y_2:\foralld(x.B)\,a_2).\foralld((x.y).C)\,(a_2,y_2)}
      {t_2 : \foralld(x.\Pi(y:B).C)\,a_2}
\]
The $t_k$s are just given by the $\d{k}_{x.\Pi(y:B).C}$s. The inverse
of this map is called $\mkpi$.

Just as $\Pi$ is not preserved up to equality by $\forall$, $\U$ is
also not preserved up to equality. There is an obvious map from
$a_2:\forall\U$ to a span: the apex is 
$\foralld(x.\El\,x)\,a_2$, the bases $\El\,(k_\U\,a_2)$, and the
legs are
$x_2:\foralld(x.\El\,x)\,a_2\vdash\d{k}_{x.\El\,x}\,a_2\,x_2$. This map has a section called $\unspan$.  The
premises of $\unspan$ are universally quantified over $k$, hence
include \emph{two} types $A_\o$ and $A_\i$, and two terms $t_\o$ and
$t_\i$.

$\Bool$, its constructors and eliminator are preserved strictly by
$\forall$ and $\apd$.

\subsection{Some Derivable Equations}

For $\Sigma$, it was enough to state that
$(\blank,\blank)$ is preserved by $\apd$, the other direction is
automatic:
\[
\apd(x.\pi_k\,t)\,a_2 = \pi_k\,\big(\apd(x.\pi_1\,t)\,a_2,\apd(x.\pi_2\,t)\,a_2\big) = \pi_k\,\big(\apd(x.(\pi_1\,t,\pi_2\,t))\,a_2\big) = \pi_k\,(\apd(x.t)\,a_2)
\]
$\R_A$ is a special case of $\ap$: $\R_A\,a = \R_A\,(a[\_\mapsto \tt]) = \ap(\_.a)\,(\R_\top\,\tt) = \ap(\_.a)\,\tt$. \\
Constant $\ap$ is constant: $\ap(\_.b)\,a_2 = \ap(\_.b[\_\mapsto\tt])\,a_2 = \ap(\_.b)\,(\ap(\_.\tt)\,a_2) = \ap(\_.b)\,\tt = 
\ap(\_.b)\,(\ap(\_.\tt)\,a_2') = \ap(\_.b[\_\mapsto\tt])\,a_2' = \ap(\_.b)\,a_2'$. \\
$k$ is distributive over $\Sigma$: $k_{\Sigma(x:A).B}\,(a_2,b_2)=(k_A\,a_2,\d{k}_{x.B}\,a_2\,b_2)$. \\
$\d{k}$ on composition: $\d{k}_{x.B[y\mapsto t]}\,c_2=\d{k}_{y.B}\,(\ap(x.t)\,c_2)$. \\
Relationship between $\d{k}$ and $\apd$: $\d{k}_{x.B}\,a_2\,(\apd(x.t)\,a_2) = t[x\mapsto k_A\,a_2]$.

\subsection{Applications}

\paragraph{Polymorphic identity.}

We revisit the ``hello world'' example of internal parametricity: a term witnessing that there is only one function with the type of the polymorphic identity. More precisely, given a function $f:\Pi(x:\Sigma (y:\U).\El\,y).\El\,(\pi_1\,x)$, a type $A$, a predicate $(x:A)\vdash P$, a term $a:A$ and a proof $p:P[x\mapsto a]$ that the predicate holds for $a$, we construct the term
\[
\apd(x.f\oldapp x)\,\big(\unspan\,A\,(\Sigma(x:A).P)\,(x.\pi_1\,x),(a,p)\big)
\]
which has type
\begin{alignat*}{10}
  & \foralld(x.\El\,(\pi_1\,x))\,\big(\unspan\,A\,(\Sigma(x:A).P)\,(x.\pi_1\,x) , (a,p)\big) & = \\
  & \foralld(x.\El\,x)\,\Big(\ap\,(x.\pi_1\,x)\,\big(\unspan\,A\,(\Sigma(x:A).P)\,(x.\pi_1\,x) , (a,p)\big)\Big) \hspace{2em} & = \\
  & \foralld(x.\El\,x)\,\big(\unspan\,A\,(\Sigma(x:A).P)\,(x.\pi_1\,x)\big) & = \\
  & \Sigma(x:A).P
\end{alignat*}
The unary version of our theory suffices (thus $k = \o$). We compute the first projection of this term.
\begin{alignat*}{10}
  & \pi_1\,\Big(\apd(x.f\oldapp x)\,\big(\unspan\,A\,(\Sigma(x:A).P)\,(x.\pi_1\,x),(a,p)\big)\Big) & = \\
  & & \hspace{-9em}(\d{k}_{x.\El\,x}\,(\unspan\,A_k\,A\,t_k)\,a = t_k[x\mapsto a]) \\
  & \d{k}_{x.\El\,x}\,\big(\unspan\,A\,(\Sigma(x:A).P)\,(x.\pi_1\,x)\big) \\
  & \hphantom{\d{k}_{x.\El\,x}\,}\Big(\apd(x.f\oldapp x)\,\big(\unspan\,A\,(\Sigma(x:A).P)\,(x.\pi_1\,x),(a,p)\big)\Big) & = \\
  & & \hspace{-9em}\text{($\d{k}$ on composition)} \\
  & \d{k}_{x.\El\,(\pi_1\,x)}\,\big(\unspan\,A\,(\Sigma(x:A).P)\,(x.\pi_1\,x),(a,p)\big) \\
  & \hphantom{\d{k}_{x.\El\,(\pi_1\,x)}\,}\Big(\apd(x.f\oldapp x)\,\big(\unspan\,A\,(\Sigma(x:A).P)\,(x.\pi_1\,x),(a,p)\big)\Big) & = \\
  & & \hspace{-9em}\text{(relation of $\d{k}$ and $\apd$)} \\
  & (f\oldapp x)[x\mapsto k_{\Sigma (y:U).\El\,y} (\unspan\,A\,(\Sigma(x:A).P)\,(x.\pi_1\,x),(a,p))] & = \\
  & & \hspace{-9em}\text{(distributivity of $k$)} \\
  & f\oldapp (c\,A,a)
\end{alignat*}
Thus the second projection provides
\[
\pi_2\,\Big(\apd(x.f\oldapp x)\,\big(\unspan\,A\,(\Sigma(x:A).P)\,(x.\pi_1\,x),(a,p)\big)\Big):P[x\mapsto f\oldapp (c\,A,a)]
\]
and we obtain the desired result by choosing the predicate $P$ to be $(x:A)\vdash \Eq_A\,a\,x$.

\paragraph{Induction for Church-encoded natural numbers.}

We can also prove results for higher order polymorphic types using $\mkpi$. As an example, we prove the induction principle for Church encoded natural numbers following \cite{wadler}.
The (uncurried) type of Church natural numbers is $$N=\Pi\Big(x:\Sigma\big(y:\Sigma(z:U).\El\,z\big).\El\,(\pi_1\,y)\ra\El\,(\pi_1\,y)\Big).\El\,\big(\pi_1\,(\pi_1\,x)\big)$$
A natural number algebra is given by a type $A$, $z_A:A$ and $s_A:A\rightarrow A$. A morphism of algebras between $(A,z_A,s_A)$ and $(B,z_B,s_B)$ is a function $f:A\rightarrow B$ such that $f\oldapp  z_A = z_B$ and $f\oldapp (s_A\oldapp n)=s_B\,(f\,n)$ for every $n$. These form a category. We define:
\begin{alignat*}{10}
& zero && :=\lambda x.\pi_2\, (\pi_1 \,x) \\
& suc && := \lambda n. \lambda x.\pi_2\, x\oldapp (n\oldapp x) \\
& ite_{(A,z_A,s_A)} && := \lambda n. n\oldapp (A,z_A,s_A)
\end{alignat*}
The induction principle that we want says that $(N,zero,suc)$ is initial in the category of algebras of $N$, $ite_A$ being the unique morphism from $N$ to the algebra $(A,z_A,s_A)$. (This is equivalent to saying that every displayed model over $N$ has a section, see \cite{DBLP:journals/pacmpl/KaposiKA19} for a generic proof.)
Using binary parametricity, we first show that $ite$ respects morphisms, that is, $f\oldapp(ite_A\oldapp n) = ite_B\oldapp n$. Assuming
\begin{alignat*}{10}
& n : N \\
& A, z_A:A, s_A:A \rightarrow A \\
& B, z_B:B, s_B:B \rightarrow B \\
& x:A\times B\vdash Q\\
& z:Q[x\mapsto (z_A,z_B)]\\
& s:\Pi\big(y:\Sigma(x:A\times B).Q\big).Q\big[x\mapsto \big(s_A\oldapp (\pi_1\,(\pi_1\, y)), s_B\oldapp (\pi_1\,(\pi_2\, y))\big)\big]
\end{alignat*}
we build
\begin{alignat*}{10}
  & ((s_A,s_B),s):\Sigma(x:A\times B).Q \rightarrow \Sigma(x:A\times B).Q\\
  & c_2 := \unspan\,A\,B\,(\Sigma(x:A\times B).Q)\,(y.\pi_1\,(\pi_1\,y))\,(y.\pi_2\,(\pi_1\,y)) : \forall\U \\
  & s^+:=\mkpi\,c_2\,s_A\,s_B\,((s_A,s_B),s) : \foralld(x.\El\,x\ra\El\,x)\,c_2.
\end{alignat*}
Finally we have
\[
\pi_2\,\Big(\apd(x.n\oldapp x)\,\big(c_2,((z_A,z_B),z),s^+\big)\Big):Q\big[x\mapsto \big(n\oldapp (c\,A,z_A,s_A),n\oldapp (c\,B,z_B,s_B)\big)\big].
\]
Now specialising this to $Q[x\mapsto(a,b)] = \Eq_B\, (f\oldapp a)\, b$ we obtain that $ite$ respects $f$.

Then we prove that the only function from $N$ to $X$ is $ite_X$ using the fact that $ite$ respects $ite_X$. This means
$ite_X\oldapp (ite_N\oldapp n) = ite_X\oldapp n$, which is the same as
$(n\oldapp (N,zero,suc))\oldapp (X,z_X,s_X) = n\oldapp (X,z_X,s_X)$. Using function extensionality we obtain $n\oldapp (N,zero,suc)= n$, which means that $ite_N$ is the identity function. We can then conclude by remarking that this gives, for any morphism $f$ from $N$ to $A$, $f\oldapp n = f\oldapp (ite_N\oldapp n) = ite_A\oldapp n$.

\paragraph{Compute $\ap$ on $\R$ using symmetry.}

One reason for needing symmetry $\S_A$ in our syntax is to compute the
$\ap(x.\R_A\,x)\,a_2$ way of turning a witness of a span $a_2:\forall
A$ into a witness of a double-span. Without symmetry, there
seems to be no way to compute the value of the closed boolean
$\ap(x.\R_\Bool\,x)\,\true$. Having symmetry and the rule
$\S_A\,(\R_{\forall A}\,a_2) = \ap(x.\R_A\,x)\,a_2$ we compute
\begin{alignat*}{10}
  & \ap(x.\R_\Bool\,x)\,\true = \S_\Bool\,(\R_{\forall\Bool}\,\true) = \S_\Bool\,(\ap(\_.\true)\,\tt) = \S_\Bool\,(\ap(x_2.\ap(\_.\true)\,x_2)\,\tt) = \\
  & \ap(x_2.\ap(\_.\true)\,x_2)\,(\S_\top\,\tt) = \ap(x_2.\ap(\_.\true)\,x_2)\,\tt = \ap(x_2.\true)\,\tt = \true.
\end{alignat*}

\section{The global theory and its presheaf model}
\label{sec:global}

In this section, we define our global theory which directly comes from
the presheaf model on BCH cubes \cite{DBLP:conf/types/BezemCH13}. We
assume basic knowledge of working with models of type theory as
categories with families (CwFs,
\cite{DBLP:journals/corr/abs-1904-00827}) and the CwF of presheaves
\cite{Hofmann97syntaxand}.

We define the global theory as a generalised algebraic theory (GAT,
\cite{DBLP:journals/pacmpl/KaposiKA19}) by saying what a model of this
theory is. First we present the core theory externally, this is the
common part of the local and global theories. This corresponds to
the core theory listed in Figure \ref{fig:core}, with the additional
structure of strict democracy which is not expressible internally.
\begin{definition}[Core theory, externally]\label{def:coreExt}
  A model of the core theory is a CwF with the following features:
  \begin{itemize}
  \item $\top$, $\Sigma$ with $\beta$ and $\eta$ and strict
    identity types $\Eq$ with uniqueness of identity proofs. This
    turns the CwF into a finite limit CwF (flCwF,
    \cite{DBLP:conf/lics/KovacsK20}).
  \item Strict democracy, that is, an operation $\K$ that turns a
    context into a type satisfying the following equations. Note that this
    includes a sort equation. (Non-strict) democracy requires
    $\Sub\,\Gamma\,\Theta = \Tm\,\Gamma\,(\K\,\Theta)$ only up to
    isomorphism.
    \begin{alignat*}{10}
      & \K : \Con\ra\Ty\,\Gamma  &&                        \K\,\diamond = \top \\                                  
      & (\K\,\Theta)[\sigma] = \K\,\Theta &&                 \K\,(\Theta\ext A) = \Sigma\,(\K\,\Theta)\,(A[\q]) \\   
      & \Sub\,\Gamma\,\Theta = \Tm\,\Gamma\,(\K\,\Theta) \hspace{7em} && (\sigma ,_{\ext} t) = (\sigma ,_{\Sigma} t)             
%      & \p = \pi_1\,\id \\ % derivable
%      & \q = \pi_2\,\id % derivable
    \end{alignat*}
    In the equation relating $\K$ of $\ext$ and $\Sigma$, $\q$ can be used as
    a substitution because of the previous equations: $\q :
    \Tm\,(\Gamma\ext\K\,\Theta)\,(\K\,\Theta[\p])$, hence $\q :
    \Tm\,(\Gamma\ext\K\,\Theta)\,(\K\,\Theta)$, and because of the
    sort equation we have $\q : \Sub\,(\Gamma\ext\K\,\Theta)\,\Theta$.
  \item $\Pi$ types given by an isomorphism $\lam : \Tm\,(\Gamma\ext
    A)\,B \cong \Tm\,\Gamma\,(\Pi\,A\,B):\app$ natural in $\Gamma$. We abbreviate
    $t\oldapp u := \app\,t[id,u]$.
  \item A hierarchy of universes {\`a} la Coquand, that is
    $\U:\Ty\,\Gamma$ with an isomorphism $\c :
    \Ty\,\Gamma\cong\Tm\,\Gamma\,\U:\El$. We don't write universe
    indices for convenience, but every construction in this paper can
    be redone in a setting where types, terms and universes are all
    indexed by levels, and $\U_i : \Ty_{i+1}\,\Gamma$.
  \item $\Bool$ with dependent elimination $\ite$ and two $\beta$
    rules (no $\eta$).
  \end{itemize}
\end{definition}
Now we list the components specific to the global theory (c.f.\ Figure
\ref{fig:cubecat}). We use {\color{brickred} brick red colour} to
distinguish from the operations of the local theory.
\begin{definition}[Global theory]\label{def:global}
  A model of the core theory is a model of the global theory if it comes
  equipped with the following structure:
  \begin{enumerate}
  \item A strict flCwF endomorphism on the model $\foralle$ that also
    preserves $\K$ strictly. We denote the four different maps by
    $\foralle_\Con$, $\foralle_\Sub$, $\foralle_\Ty$ and
    $\foralle_\Tm$ or just $\foralle$. Strict preservation of $\K$
    means $\foralle(\K\,\Theta) = \K\,(\foralle\Theta)$ and
    $\foralle_\Sub\,\{\Gamma\}\{\Theta\} =
    \foralle_\Tm\,\{\Gamma\}\{\K\,\Theta\}$.
  \item Natural transformations $\ke$ from $\foralle$ to the identity
    functor on the CwF, for $k = \o$ and $k = \i$. This means for each
    $\Gamma:\Con$, a substitution $\ke_\Gamma :
    \Sub\,(\foralle\Gamma)\,\Gamma$ with
    $\ke_\Gamma\circ\foralle\sigma = \sigma\circ\ke_\Delta$. We define
    the action of $\ke$ on types as follows.
    \begin{alignat*}{10}
      & \ke_A : \Tm\,(\foralle\Gamma\ext\foralle A)\,(A[\ke_\Gamma\circ\p]) \\
      & \ke_A := \q[\ke_{\Gamma\ext A}]
    \end{alignat*}
  \item A natural transformation $\Re$ from the identity to $\foralle$.
  \item A natural transformation $\Se$ from $\foralle\circ\foralle$ to
    itself.
  \item The following five equations relating the above three natural transformations:
    \begin{alignat*}{10}
      \ke_\Gamma\circ\Re_\Gamma & = \id_\Gamma \\
      \ke_{\foralle\Gamma}\circ\Se_\Gamma & = \foralle\ke_\Gamma \\
      \Se_\Gamma\circ\Re_{\foralle\Gamma} & = \foralle\Re_\Gamma \\
      \Se_\Gamma\circ\Se_\Gamma & = \id_{\foralle(\foralle\Gamma)} \\
      \Se_{\foralle\Gamma}\circ\foralle\Se_\Gamma\circ\Se_{\foralle\Gamma} & = \foralle\Se_\Gamma\circ\Se_{\foralle\Gamma}\circ\foralle\Se_\Gamma
    \end{alignat*}
  \item The map from $\foralle(\Pi\,A\,B)$ to compatible
    $(\Pi\,A\,B[\ke])$s and $(\Pi\,(\foralle A)\,(\foralle B))$ has an inverse $\mkpie$. Precisely we require the operation $\mkpie$ with the following equations.
    \begin{alignat*}{10}
      & \mkpie :{} && \big(\sigma:\Sub\,\Delta\,(\foralle\Gamma)\big)\big(t_k : \Tm\,\Delta\,(\Pi\,A\,B[\ke_\Gamma\circ\sigma])\big)\big(t:\Tm\,\Delta\,(\Pi\,(\foralle A)\,(\foralle B)[\sigma])\big)\ra \\
      & && \big(t_k[\p]\oldapp (\ke_A[\sigma\uparrow]) = \ke_B[\sigma\uparrow,\app\,t]\big) \ra \Tm\,\Delta\,(\foralle(\Pi\,A\,B)[\sigma]) \\
      & \rlap{$\mkpie\,\sigma\,t_k\,t[\rho] = \mkpie\,(\sigma\circ\rho)\,(t_k[\rho])\,(t[\rho])$} \\
      & \rlap{$\ke_{\Pi\,A\,B}[\sigma,\mkpie\,\sigma\,t_k\,t] = t_k$} \\
      & \rlap{$\lam\,(\foralle(\app\,\q))[\sigma,\mkpie\,\sigma\,t_k\,t] = t$} \\
      & \rlap{$\mkpie\,\sigma\,(\ke_{\Pi\,A\,B}[\sigma,t_2])\,(\lam\,(\foralle(\app\,\q))[\sigma,t_2]) = t_2$}
    \end{alignat*}
    When we quantify over a subscripted variable such as $t_k$, we mean to quantify over \emph{two} variables $t_\o$ and $t_\i$ with types obtained by substituting $k=\o,\i$ in the type of $t_k$.
  \item The map from $\foralle\U$ to a span has a section called $\unspane$. Precisely we require:
    \begin{alignat*}{10}
      & \unspane : (A_k\,A:\Ty\,\Gamma)\ra\Tm(\Gamma\ext A)\,(A_k[\p])\ra\Tm\,\Gamma\,(\foralle\U[\epsilon]) \\
      & \unspane\,A_k\,A\,t_k[\sigma] = \unspane\,(A_k[\sigma])\,(A[\sigma])\,(t_k[\sigma\uparrow]) \\
      & \El\,(\ke_\U)[\epsilon,\unspane\,A_k\,A\,t_k] = A_k \\
      & \foralle(\El\,\q)[\epsilon,\unspane\,A_k\,A\,t_k] = A \\
      & \ke_{\El\,\q}[(\epsilon,\unspane\,A_k\,A\,t_k)\uparrow] = t_k
    \end{alignat*}
    We do not require the other roundtrip equation.
  \item $\foralle$ preserves $\Bool$ strictly.
  \end{enumerate}
\end{definition}

We define the BCH cube category as depicted in Figure \ref{fig:cubecat}.
\begin{definition}[Cube category $\square$]\label{def:square}
  The objects of the category are natural numbers, the morphisms are
  generated by the following quotient inductive set. First of all we
  have a category, that is, composition and $\id$ with the categorical
  laws, then we have the following constructors and equality
  constructors. When we write $k$ we always mean both a copy for $\o$
  and a copy for $\i$.
  \begin{alignat*}{10}
    & \suc && {}:\square(J,I)\ra\square(1+J,1+I) \hspace{2em} && \suc\,(f\circ g) = \suc\,f\circ\suc\,g \hspace{2em} && \suc\,\id = \id \\
    & k_I && {}:\square(I,1+I) && k_I\circ f = \suc\,f\circ k_J \\
    & \R_I && {}:\square(1+I,I) && \R_I\circ\suc\,f = f\circ\R_J \\
    & \S_I && {}:\square(2+I,2+I) && \S_I\circ\suc\,(\suc\,f) = \suc\,(\suc\,f)\circ\S_J \\
    & && && \R_I\circ k_I = \id_I \\
    & && && \S_I\circ k_{1+I} = \suc\,k_I \\
    & && && \R_{1+I}\circ\S_I = \suc\,\R_I \\
    & && && \S_I\circ\S_I = \id_{2+I} \\
    & && && \S_{1+I}\circ\suc\,\S_I\circ\S_{1+I} = \suc\,\S_I\circ\S_{1+I}\circ\suc\,\S_I
  \end{alignat*}
\end{definition}

\paragraph{Comment on the definition of $\square$.}

Another way to describe this category is that it is the free symmetric
semicartesian strict monoidal category over a cylinder.  In
particular, the final equation on $\S$ is the ``braid equation'',
which together with naturality forms a presentation of the symmetric
group; thus the automorphisms of $I$ are the permutations of an
$I$-element set.  This category also has a named variant, in which the
objects are finite sets and a morphism from $J$ to $I$ is a function
from $J$ to $I+\{\o,\i\}$ that is injective on the subset of $J$ that is not
mapped to $\mathsf{inr}$.

$\S_I$ allows us to swap the first two dimensions, but we would like
to swap the first with any other dimension, on both directions.
\begin{definition}[Generalised symmetries in $\square$]
  We define the following morphisms by induction on natural numbers.
  \begin{alignat*}{10}
    & (\sym_{I_0,1}|\id_{I_1}) && {}: \square(I_0+1+I_1,1+I_0+I_1) \hspace{5em} && (\sym_{1,I_0}|\id_{I_1}) && {}: \square(1+I_0+I_1,I_0+1+I_1) \\
    & (\sym_{0,1}|\id_{I_1}) && {}:= \id_{1+I_1} && (\sym_{1,0}|\id_{I_1}) && {}:= \id_{1+I_1} \\
    & (\sym_{1+I_0,1}|\id_{I_1}) && {}:= \S_{I_0+I_1}\circ\suc\,(\sym_{I_0,1}|\id_{I_1}) && (\sym_{1,1+I_0}|\id_{I_1}) && {}:= \suc\,(\sym_{1,1+I_0}|\id_{I_1})\circ\S_{I_0+I_1}
  \end{alignat*}
\end{definition}
By induction we prove that the generalised symmetries form an isomorphism in $\square$:
\[
(\sym_{I_0,1}|\id_{I_1})\circ(\sym_{1,I_0}|\id_{I_1}) = \id_{1+I_0+I_1} \hspace{2em}\text{and}\hspace{2em}(\sym_{1,I_0}|\id_{I_1})\circ(\sym_{I_0,1}|\id_{I_1}) = \id_{I_0+1+I_1}
\]
We would like to characterise the morphisms in $\square$ without
equations, so that we can have a powerful case analysis (which will be needed for defining the $\foralle$-preservation of $\Pi$ and $\U$). We will need
a case analysis on a morphism $f : \square(J,1+I)$ to know where the
$1$ in the codomain is coming from: either it is coming from a face
map $k_I$ and $f = k_I\circ f'$ or $J = J_0+1+J_1$ for some $J_0$ and
$J_1$ and the $1$ in the codomain comes from the $1$ in the domain,
with the rest being mapped by some $f':\square(J_0+J_1,I)$. We
generalise so that the $1$ in the codomain can be in the middle.
\begin{problem}[Case analysis on morphisms in $\square$]\label{con:case}
  For a morphism $f:\square(J,I_0+1+I_1)$ either
  \begin{enumerate}
  \item[(i)] there is a $k$ and an $f':\square(J,I_0+I_1)$ such that $f = (\sym_{1,I_0}|\id_{I_1}) \circ k_{I_0+I_1}\circ f'$,
  \item[(ii)] or there are $J_0$, $J_1$ with $J = J_0+1+J_1$ and $f':\square(J_0+J_1,I_0+I_1)$ such that
    $f = (\sym_{1,I_0}|\id_{I_1})\circ\suc\,f'\circ(\sym_{J_0,1}|\id_{J_1})$.
  \end{enumerate}
\end{problem}
\begin{proof}[Construction]
  The first case means that the $1$ in the codomain of the morphism
  comes from a face map $k$, the second case means that the $1$ in the
  codomain comes from a $1$ in the domain.
  
  We perform induction by building a displayed model of the quotient
  inductive set in Definition \ref{def:square}, and then using the
  induction principle of the quotient inductive set
  \cite{DBLP:journals/pacmpl/KaposiKA19}. Displayed morphisms are given by
  \begin{alignat*}{10}
    & \square^\bullet(J,I)\,f := \{I_0\,I_1:\square\}\{I = I_0+1+I_1\} \ra \big(\exists k.(f':\square(J,I_0+I_1))\times f = (\sym_{1,I_0}|\id_{I_1}) \circ k_{I_0+I_1}\circ f'\big) + \\
    & \hspace{2em}\big((J_0:\square)\times(J_1:\square)\times(f':\square(J_0+J_1,I_0+I_1))\times f = (\sym_{1,I_0}|\id_{I_1})\circ\suc\,f'\circ(\sym_{J_0,1}|\id_{J_1}) \big).
  \end{alignat*}
  Displayed composition is given by
  \begin{alignat*}{10}
    & \inl\,(k,f')\circ^\bullet g^\bullet && {}:= \inl\,(k,f'\circ g) \\
    & \inr\,(J_0,J_1,f')\circ^\bullet\inl\,(l,g') && {}:= \inl\,(l,f'\circ g') \\
    & \inr\,(J_0,J_1,f')\circ^\bullet\inr\,(K_0,K_1,g') && {}:= \inr\,(K_0,K_1,f'\circ g').
  \end{alignat*}
  We would like to define $f^\bullet\circ^\bullet g^\bullet$ which
  says our induction motive for $f\circ g$. We match on the induction
  hypothesis for $f$ denoted $f^\bullet$: if it is an $\inl$, that is,
  $f = (\sym_{1,I_0}|\id_{I_1}) \circ k_{I_0+I_1}\circ f'$, then we
  have $f\circ g = (\sym_{1,I_0}|\id_{I_1}) \circ k_{I_0+I_1}\circ
  f'\circ g$, so we still know that the $1$ in the codomain of $f\circ
  g$ comes from $k$. If $f^\bullet$ is an $\inr$ (that is, the $1$ in
  the codomain comes from a $1$ in the domain of $f$), then we have to
  match on $g^\bullet$ to learn whether that $1$ in the codomain of
  $g$ comes from $k$ or from the domain of $g$. We then check that
  $\blank\circ^\bullet\blank$ satisfies associativity.

  For the rest of the constructors of $\square$, we list the
  computational parts of the displayed model.
  \begin{alignat*}{10}
    & \id^\bullet\,\{I_0\}\{I_1\} && {}:= \inr\,(I_0,I_1,\id) \\
    & \suc^\bullet\,f^\bullet\,\{0\}\{I\} && {}:= \inr\,(0,J,f) \\
    & \suc^\bullet\,f^\bullet\,\{1+I_0\}\{I_1\} && {}:= \inl\,(k,\suc\,f') \text{, if } f^\bullet\,\{I_0\}\{I_1\} = \inl\,(k,f') \\
    & \suc^\bullet\,f^\bullet\,\{1+I_0\}\{I_1\} && {}:= \inl\,(1+J_0,J_1,\suc\,f') \text{, if } f^\bullet\,\{I_0\}\{I_1\} = \inr\,(J_0,J_1,f') \\
    & k^\bullet\,\{0\}\{I\} && {}:= \inl\,(k,\id_I) \\
    & k^\bullet\,\{1+I_0\}\{I_1\} && {}:= \inr\,(I_0,I_1,k_{I_0+I_1}) \\
    & \R^\bullet\,\{I_0\}\{I_1\} && {}:= \inr\,(1+I_0,I_1,\R_{I_0+I_1}) \\
    & \S^\bullet\,\{0\}\{1+I\} && {}:= \inr\,(1,I,\id_{1+I}) \\
    & \S^\bullet\,\{1\}\{I\} && {}:= \inr\,(0,1+I,\id_{1+I}) \\
    & \S^\bullet\,\{2+I_0'\}\{I_1'\} && {} := \inr\,(2+I_0',I_1',\S_{I_0'+I_1'})
  \end{alignat*}
  The displayed versions of the equations all hold. The displayed
  model induces a dependent function from $\square(J,I_0+1+I_1)$ to
  the sum type by the induction principle of the quotient inductive
  set $\square$.
\end{proof}
This case analysis will be used in parts (6) and (7) of Problem
\ref{con:psh_global}.

\begin{construction}[Presheaf model of the core theory]
We recall the computational parts of the presheaf model over $\square$
\cite{Hofmann97syntaxand} providing a model of the core theory
(Definition \ref{def:coreExt}).

A $\Gamma:\Con$ is a presheaf, that is a family of sets
$\Gamma:\square\ra\Set$ together with reindexing $\gamma_I[f]_\Gamma :
\Gamma\,J$ for $\gamma_I:\Gamma\,I$ and $f:\square(J,I)$ such that
$\gamma_I[f\circ g] = \gamma_I[f][g]$ and $\gamma_I[\id] =
\gamma_I$. A $\sigma:\Sub\,\Delta\,\Gamma$ is a polymorphic function
$\sigma:\{I:\square\}\ra\Delta\,I\ra\Gamma\,I$ with
$(\sigma\,\delta_I)[f]_\Gamma = \sigma\,(\delta_I[f]_\Delta)$. A type
$A:\Ty\,\Gamma$ is a dependent presheaf, that is
$A:(I:\square)\ra\Gamma\,I\ra\Set$ equipped with $\blank[\blank]_A :
A\,I\,\gamma_I\ra(f:\square(J,I))\ra A\,J\,(\gamma_I[f])$ with two
equations. Finally a term is a dependent natural tansformation $t :
(\gamma_I:\Gamma\,I)\ra A\,I\,\gamma_I$ with $t\,\gamma_I[f]_A =
t\,(\gamma_I[f]_\Gamma)$. Context extension and the types $\top$,
$\Sigma$, $\Eq$ are defined pointwise and we have strict democracy by
$\K\,\Theta\,I\,\gamma_I = \Theta\,I$ and $\theta_I[f]_{\K\,\Theta} =
\theta_I[f]_{\Theta}$. The Yoneda embedding is a functor from
$\square$ to presheaves over $\square$, we denote it by
$\y:\square\ra\Con$ and $\y:\square(J,I)\ra\Sub\,(\y\,J)\,(\y\,I)$,
and we have the Yoneda lemma $\yl_\Gamma :
\Gamma\,I\ra\Sub\,(\y\,I)\,\Gamma$ defined by $\yl_\Gamma\,\gamma_I\,f
= \gamma_I[f]_{\Gamma}$. Function space is defined as
$\Pi\,A\,B\,I\,\gamma_I = \Tm\,(\y\,I\ext
A[\yl\,\gamma_I])\,(B[\yl\,\gamma_I\uparrow])$ with $t[f]_{\Pi\,A\,B}
= t[\y\,f\uparrow]$. $\lam\,t\,\gamma_I = t[\yl\,\gamma_I\uparrow]$
and $\app\,t\,(\gamma_I,a_I) = t\,\gamma_I\,(\id_I,a_I)$. The universe
is defined as $\U\,I\,\gamma_I = \Ty\,(\y\,I)$ with $a[f]_\U :=
a[\y\,f]$. Decoding is $\El\,a\,I\,\gamma_I = a\,\gamma_I\,I\,\id_I$
with $u_I[f]_{\El\,a} = u_I[f]_{a\,\gamma_I}$, and encoding is
$\c\,A\,\gamma_I = A[\yl\,\gamma_I]$. We have $\El\,a[\yl\,\gamma_I] =
a\,\gamma_I$. Bool is pointwise.
\end{construction}

\begin{problem}[Presheaf model of the global theory]\label{con:psh_global}
The presheaf model extends to the operations and equations of the
global theory (Definition \ref{def:global}). This justifies the notion
of the global theory, which was in turn extracted from this exact
presheaf model.
\end{problem}
\begin{proof}[Construction]
  Recall the two diagrams of Figure \ref{fig:cubecat}. We define the components in order.
  \begin{enumerate}
  \item $\foralle$ is defined as precomposition with $\suc$:
    \begin{alignat*}{10}
      & \foralle\Gamma\,I && {}:= \Gamma\,(1+I)                         &&                \foralle A\,I\,\gamma_{1+I} && {}:= A\,(1+I)\,\gamma_{1+I} \hspace{3em} &&  \foralle\sigma\,\delta_{1+I} && := \sigma\,\delta_{1+I} \\
      & \gamma_{1+I}[f]_{\foralle\Gamma} && {}:= \gamma_{1+I}[\suc\,f]_{\Gamma} \hspace{3em} &&        a_{1+I}[f]_{\foralle A} && {}:= a_{1+I}[\suc\,f]_{A}                   &&  \foralle t\,\gamma_{1+I} && := t\,\gamma_{1+I}              
    \end{alignat*}
    This preserves the flCwF structure strictly, e.g.\ we have
    $\foralle(\Gamma\ext A)\,I = (\Gamma\ext A)\,(1+I) =
    (\gamma_{1+I}:\Gamma\,(1+I))\times A\,(1+I)\,\gamma_{1+I} =
    (\gamma_{1+I}:\foralle\Gamma\,I)\times\foralle A\,I\,\gamma_{1+I}
    = (\foralle\Gamma\ext\foralle A)\,I$, strict preservation of $\K$
    is by $\foralle(\K\,\Theta)\,I\,\gamma_{1+I} =
    \K\,\Theta\,(1+I)\,\gamma_{1+I} = \Theta\,(1+I) =
    \foralle\Theta\,I = \K\,(\foralle\Theta)\,I\,\gamma_{1+I}$ and by
    the fact that $\foralle_\Sub$ and $\foralle_\Tm$ coincide.
  \item $\ke$ is defined as $\ke_\Gamma\,\gamma_{1+I} :=
    \gamma_{1+I}[k_I]_\Gamma$. This is natural as $(\ke_\Gamma\circ\foralle\sigma)\,\delta_I = (\sigma\,\delta_I)[k_I]{\Gamma} \overset{\sigma \text{ natural}}{=} \sigma\,(\delta_I[k_I]{\Delta}) = (\sigma\circ\ke_\Delta)\,\delta_I$
  \item $\Re_\Gamma\,\gamma_I := \gamma_{I}[\R_I]_\Gamma$
  \item $\Se_\Gamma\,\gamma_{2+I} := \gamma_{2+I}[\S_I]_\Gamma$
  \item The five equations follow from their corresponding equations
    in $\square$. For example, $(\ke_\Gamma\circ\Re_\Gamma)\,\gamma_I = \gamma_I[\Re_I]_{\Gamma}[\ke_I]_{\Gamma} = \gamma_I[\Re_I\circ\ke_I]_{\Gamma} = \gamma_I[\id_I]_{\Gamma} = \gamma_I = \id_\Gamma\,\gamma_I$.
  \item $\mkpie\,\sigma\,t_k\,t\,\delta_I$ is in
    $\foralle(\Pi\,A\,B)\,I\,(\sigma\,\delta_I)$, that is, in
    $\Tm\,(\y(1+I)\ext A[\yl_\Gamma\,(\sigma\,\delta_I)])\,(B[(\yl_\Gamma\,(\sigma\,\delta_I))\uparrow])$. We do case
    analysis (Problem \ref{con:case}) on the input morphism given
    by Yoneda $\y\,(1+I)$:
    \begin{alignat*}{10}
      & \mkpie\,\sigma\,t_k\,t\,\delta_I\,\{J\}\,        && (k_I\circ f,a_J) && {}:= \app\,t_k\,(\delta_I[f]_\Delta, a_J) \\
      & \mkpie\,\sigma\,t_k\,t\,\delta_I\,\{J_0+1+J_1\}\,&& \big((\suc\,f\circ(\sym_{J_0,1}|\id_{J_1})),a_{J_0+1+J_1}\big) && {}:= \\
      & \rlap{$\hspace{14em}\app\,t\,(\delta_I[f]_{\Delta}, a_{J_0+1+J_1}[(\sym_{1,J_0}|id_{J_1})]_A)[\sym_{J_0,1}|\id_{J_1}]_B$}
    \end{alignat*}
    If the morphism selects a projection $k$, then we use the map at
    the base of the span $t_k$. If the morphism selects an existing
    dimension ($1$ in the middle of $J_0+1+J_1$), then we use the
    function at the apex $t$, but we have to apply symmetry to the
    input to make it well-typed. $\app\,t$ requires an input in
    $\foralle A[\sigma]\,(J_0+J_1) (\delta_I[f]) =
    A\,(1+J_0+J_1)\,(\sigma\,(\delta_I[f]))$, but our $a_{J_0+1+J_1}$
    is in
    $A\,(J_0+1+J_1)\,(\sigma\,\delta_I[\suc\,f\circ(\sym_{J_0,1|id_{J_1}})])$. Then
    we have to apply symmetry at the output again:
    $\app\,t\,(\delta_I[f]_{\Delta},
    a_{J_0+1+J_1}[(\sym_{1,J_0}|id_{J_1})]_A)$ is in
    $B\,(1+J_0+J_1)\,(\sigma\,(\delta_I[f]),a_{J_0+1+J_1}[\sym_{1,J_0}|\id_{J_1}])$,
    but we need an element of
    $B\,(J_0+1+J_1)\,(\sigma\,\delta_I[\suc\,f\circ(\sym_{J_0,1}|\id_{J_1})],
    a_{J_0+1+J_1})$.

    This $\mkpie$ operation is natural in three different senses.
    From naturality of $t_k$ and $t$ and from the equation
    $t_k[\p]\oldapp (\ke_A[\sigma\uparrow]) =
    \ke_B[\sigma\uparrow,\app\,t]$, we get
    $\mkpie\,\sigma\,t_k\,t\,\delta_I\,(f,a_J)[g]_{B} =
      \mkpie\,\sigma\,t_k\,t\,\delta_I\,(f\circ g,a_J[g]_A)$.
    The naturalities
    $\mkpie\,\sigma\,t_k\,t\,\delta_I[f]_{\foralle(\Pi A B)} =
    \mkpie\,\sigma\,t_k\,t\,(\delta_I[f]_{\Delta})$ and
    $\mkpie\,\sigma\,t_k\,t[\rho] =
    \mkpie\,(\sigma\circ\rho)\,(t_k[\rho])\,(t[\rho])$ are also easy
    consequences.

    The three equations which express the roundtrips follow by
    unfolding the definitions.
  \item $\unspane\,A_k\,A\,t_k\,\gamma_I$ is in
    $\foralle\U\,I\,\gamma_I = \Ty\,(\y\,(1+I))$. We define this type
    by case analysis on the morphism given by Yoneda.
    \begin{alignat*}{10}
      & \unspane\,A_k\,A\,t_k\,\gamma_I\,J\,&&(k_I\circ f) && := A_k\,J\,(\gamma_I[f]) \\
      & \unspane\,A_k\,A\,t_k\,\gamma_I\,(J_0+1+J_1)\,&& (\suc\,f\circ(\sym_{J_0,1}|\id_{J_1})) && := A\,(J_0+J_1)\,(\gamma_I[f])
    \end{alignat*}
    The restriction operation on this type is given by restriction in
    $A_k$ or $A$, unless the morphism takes us from $A$ to $A_k$: in
    this case we apply $t_k$.
    \begin{alignat*}{10}
      & {a_k}_J[g]_{\unspane\,A_k\,A\,t_k\,\gamma_I} && {}:= {a_k}_J[g]{A_k} \\
      & a_{J_0+J_1}[(\sym_{1,J_0}|\id_{J_1}) \circ k_{J_0+J_1}\circ g]_{\unspane\,A_k\,A\,t_k\,\gamma_I} && {}:= t_k\,(\gamma_I[f\circ g], a_{J_0+J_1}[g]_A) \\
      & a_{J_0+J_1}[(\sym_{1,J_0}|\id_{J_1})\circ\suc\,g\circ(\sym_{K_0,1}|\id_{K_1})]_{\unspane\,A_k\,A\,t_k\,\gamma_I} && {}:= a_{J_0+J_1}[g]_A
    \end{alignat*}
    In the second case we expected a result in the set
    \begin{alignat*}{10}
      & \unspane\,A_k\,A\,t_k\,\gamma_I\,K\,\big(\suc\,f\circ(\sym_{J_0,1}|\id_{J_1})\circ(\sym_{1,J_0}|\id_{J_1})\circ k_{J_0+J_1}\circ g\big) = \\
      & \unspane\,A_k\,A\,t_k\,\gamma_I\,K\,(\suc\,f\circ k_{J_0+J_1}\circ g) = \\
      & \unspane\,A_k\,A\,t_k\,\gamma_I\,K\,(k_{I}\circ f\circ g) = \\
      & A_k\,K\,(\gamma_I[f\circ g]),
    \end{alignat*}
    this is why we used the leg of the span $t_k$.

    The equation $\El\,(\ke_\U)[\epsilon,\unspane\,A_k\,A\,t_k] = A_k$
    holds by definition of $\ke_\U$ which becomes
    $\blank[k_I]_\U$. The equation
    $\foralle(\El\,\q)[\epsilon,\unspane\,A_k\,A\,t_k] = A$ also
    holds by definition, here we rely on the fact that $\El$ is
    applies the $\id$ morphism to the code for the type
    inside. Finally,
    $\ke_{\El\,\q}[(\epsilon,\unspane\,A_k\,A\,t_k)\uparrow] = t_k$
    comes from the definition of restriction for
    $\unspane\,A_k\,A\,t_k\,\gamma_I$ which applies $t_k$ for
    morphisms ending in $k$.

    Note that this model \emph{does not} justify the other roundtrip,
    so $\unspane$ is not an isomorphism, just a section. Given an $a_2
    : \Tm\,\Gamma\,(\foralle\U[\epsilon])$, we don't necessarily have
    \[
    \unspane\,(\El\,\ke_U[\epsilon,a_2])\,(\foralle(\El\,\q)[\epsilon,a_2])\,(\ke_{\El\,q}[(\epsilon,a_2)\uparrow]) = a_2.
    \]
    Given a $\gamma_I:\Gamma\,I$ we have that $a_2\,\gamma_I$ is in $\Ty\,(\y\,(1+I))$, but
    \begin{alignat*}{10}
      & \unspane\,(\dots)\,(\foralle(\El\,\q)[\epsilon,a_2])\,(\dots)\,\gamma_I\,(J_0+1+J_1)\,(\suc\,f\circ(\sym_{J_0,1}|\id_{J_1})) = \\
      & \foralle(\El\,\q)[\epsilon,a_2]\,(J_0+J_1)\,(\gamma_I[f]) = \\
      & \El\,q[\epsilon,a_2]\,(1+J_0+J_1)\,(\gamma_I[f]) = \\
      & a_2\,(\gamma_I[f]_\Gamma)\, (1+J_0+J_1)\,\id = \\
      & (a_2\,\gamma_I)[f]_{\foralle\U}\, (1+J_0+J_1)\,\id = \\
      & (a_2\,\gamma_I)[\suc\,f]_\U \,(1+J_0+J_1)\,\id = \\
      & (a_2\,\gamma_I)[\y\,(\suc\,f)]\, (1+J_0+J_1)\,\id = \\
      & a_2\,\gamma_I \,(1+J_0+J_1)\,(\suc\,f) \neq \\
      & a_2\,\gamma_I\,(J_0+1+J_1)\,(\suc\,f\circ(\sym_{J_0,1}|\id_{J_1})).
    \end{alignat*}
    Note the inequality in the penultimate line.
  \item Preservation of booleans is trivial.\qedhere
  \end{enumerate}
\end{proof}

\section{Isomorphism of the local and global syntaxes}
\label{sec:iso}

In this section we construct a model of the local theory using the
global syntax and show that it is isomorphic to the local syntax. And
similarly, we construct a model of the global theory using the local
syntax.

\subsection{Defining the Local Syntax Using the Global Syntax}
\label{sec:usingglobal}

The types in the local theory correspond to the types in a fixed context $\Gamma$ of the core theory. To define the local $\forall$ from the global $\foralle$, we apply $\foralle$ and apply the substitution that goes back to $\Gamma$:
\begin{align*}
    & \forall : \Ty\,\Gamma \ra \Ty\,\Gamma\\
    & \forall A = (\foralle A)[\Re_{\Gamma}]
\end{align*}
All the other operations follow the same idea:
\begin{alignat*}{10}
    & \ap : \Tm\,(\Gamma\ext A)\,(B[\p]) \ra \Tm\,(\Gamma\ext \forall A)\,(\forall B[\p]) \hspace{2.5em} && k_{\blank} : (A:\Ty\,\Gamma)\ra\Tm\,(\Gamma\ext\forall A)\,(A[\p])\\
    & \ap\,t := (\foralle t)[\Re_{\Gamma}\uparrow] && k_A := \q[\ke_{\Gamma\ext A}][\Re_{\Gamma}\uparrow]\\
    & \foralld : \Ty\,(\Gamma\ext A)\,B \ra \Ty\,(\Gamma\ext \forall A)\,(\forall B) && \R_{\blank} : (A:\Ty\,\Gamma)\ra\Tm\,(\Gamma\ext A)\,(\forall A[\p]) \\
    & \foralld B := (\foralle B)[\Re_{\Gamma}\uparrow] && \R_A := \q[\Re_{\Gamma\ext A}]\\
    & \apd : \Tm\,(\Gamma\ext A)\,B \ra \Tm\,(\Gamma\ext \forall A)\,(\foralld B) && \S_{\blank} : (A:\Ty\,\Gamma)\ra\Tm\,(\Gamma\ext\forall(\forall A))\,(\forall(\forall A)[\p])\\
    & \apd\,t := (\foralle t)[\Re_{\Gamma}\uparrow] && \S_A := \q[\Se_{\Gamma\ext A}][\Re_{\foralle \Gamma}\circ \Re_{\Gamma}\uparrow]
\end{alignat*}
The equations of the local theory follow from the corresponding equations in the global theory and some naturality principles. We give an example:
\begin{align*}
    k_A[\p,\R_A] & = \q[\ke_{\Gamma\ext A}][\Re_{\Gamma}\uparrow][\p,\q[\Re_{\Gamma\ext A}]]\\
                  & = \q[\ke_{\Gamma\ext A}][\Re_{\Gamma}\circ \p,\q[\Re_{\Gamma\ext A}]]\\
                  & = \q[\ke_{\Gamma\ext A}][\p\circ \Re_{\Gamma\ext A},\q[\Re_{\Gamma\ext A}]]  & \textnormal{(naturality of $\Re$)}\\
                  & = \q[\ke_{\Gamma\ext A}][\Re_{\Gamma\ext A}]\\
                  & = \q & \textnormal{(corresponding global equation)}
\end{align*}
$\unspan$ is simply $\unspane$ and $\mkpi\,t_k\,t := \mkpie\,(\Re_\Gamma\uparrow)\,t_k\,t$.

\subsection{Defining the Global Syntax Using the Local Syntax}
\label{sec:usinglocal}

We have an operation $\forall$ on types and we want to extend it to contexts. Since the syntax is contextual, it is natural to do an induction to define $\foralle$. In fact, since we have strict democracy, we have another natural way to procede which is to apply $\forall$ on the type corresponding to the context, so we need a mutual induction to make the link between these two, by defining an isomorphism between them. 
\begin{alignat*}{10}
  & \foralle : \Con &&{}\ra\Con && \foralle^\cong : (\Gamma:\Con)&&{}\ra \foralle\Gamma \cong \diamond\ext\forall(\K\,\Gamma) \\
  & \foralle \diamond && {}:= \diamond && \foralle^\cong \diamond &&{}:= (\epsilon,\tt)\\
  & \foralle (\Gamma\ext A) &&{}:= \foralle \Gamma\ext\foralld(A[\q])[\foralle^\cong \Gamma] \hspace{5em} && \foralle^\cong (\Gamma\ext A) &&{}:= (\epsilon, (\q[\foralle^\cong \Gamma \circ \p],\q))\\
  & && && {\foralle^\cong}^{-1} (\Gamma\ext A)&&{}:= ({\foralle^\cong}^{-1} \Gamma\circ(\p,\pi_1\,\q),\pi_2\,\q)
\end{alignat*}
Here we used $\q:\Tm\,(\diamond\ext \K\,\Gamma)\,(\K\,\Gamma)$ as a substitution $\Sub\,(\diamond\ext \K\,\Gamma)\,\Gamma$, as a consequence of strict democracy. In fact, it is one direction of an isomorphism:
$$\q:\diamond\ext \K\,\Gamma \cong \Gamma : (\epsilon,\id) $$
We also have to do an induction for substitutions.
\begin{alignat*}{10}
& \rlap{$\foralle : \{\Gamma:\Con\}(\sigma:\Sub\,\Delta\,\Gamma)\ra \Sub\,(\foralle\Delta)\,(\foralle\Gamma)$} \\
& \foralle\,\{\diamond\}\,&& \sigma &&{}:= \epsilon  \\
& \foralle\,\{\Gamma\ext A\}\,&&(\sigma,t)&&{}:=\big(\foralle\,\{\Gamma\}\,\sigma,\apd\,(t[\q])[\foralle^\cong \Delta]\big)
\end{alignat*}
We can also prove a naturality property for $\foralle^\cong$:
$$\foralle^= : (\Gamma:\Con)(\sigma:\Sub\,\Delta\,\Gamma)\ra \foralle^\cong\,\Gamma\circ\foralle\sigma=(\p,\ap(\sigma[\q]))\circ\foralle^\cong\,\Delta$$

\noindent The definition of $\foralle$ on types and terms is enforced:
$$\foralle A:= \foralld(A[\q])[\foralle^\cong \Gamma] \hspace{10em} \foralle t := \apd\,(t[\q])[\foralle^\cong \Delta]$$

\noindent Finally we define:
$$\ke_\Gamma:= k_{\K\,\Gamma}[\foralle^\cong \Gamma]
    \qquad\qquad \Re_\Gamma:= {\foralle^\cong}^{-1}\Gamma\circ(\epsilon,\ap\,\p[\id,\tt])
    \qquad\qquad \Se_\Gamma:= \q\circ(\p,\S_{\K\,\Gamma})\circ (\epsilon,\id)
$$
and 
%$$\mkpie\,\sigma\,t_k\,t := (\mkpi \{\top\}\{A[k_{\Gamma}\circ\sigma\circ\p]\}\{B[k_{\Gamma}\circ\sigma\circ\p^2,\q]\}\,t_k[\p]\,t[\p])[\id,\tt]$$
$\mkpie\,\sigma\,t_k\,t := \mkpi\,(\q[\foralle^\cong\,\Gamma\circ\sigma])\,t_k\,t$.
Once again, $\unspane$ is straightforward, and the equations come from the corresponding equations in the local theory.

\subsection{Roundtrips}
\label{sec:roundtrips}

The 9 operations of the local syntax coincide with the same ones after
globalising and localising them. The proofs essentially rely on manipulations on substitutions and naturality properties.\\
\\
In addition, the 9 operations of the global syntax coincide with the same ones
after localising and globalising them. To prove that the $\foralle'$ that we define by induction is the same as the initial $\foralle$, we do a mutual induction proving that $\foralle^\cong \Gamma=(\epsilon,\id)$. The proof for substitutions is also an induction, the others rely on naturality. 

\subsection{Isomorphism}
\label{sec:iso_summary}

Subsection \ref{sec:usingglobal} says that the global syntax $\e\Syn$
is also a model of the local theory. Thus we obtain a map $\alpha$
from the local syntax to this global model. Note that it is identity
on the core calculus (Figure \ref{fig:core} or Definition
\ref{def:coreExt}). Similarly, Subsection \ref{sec:usinglocal}
decorates the local syntax with a model of the global syntax providing
a morphism $\beta$ from $\e\Syn$ to the local syntax. By induction by
$\Syn$ and $\e\Syn$, we prove the compositions of $\alpha$ and $\beta$
are identities. The only nontrivial cases of this induction are
handled in Subsection \ref{sec:roundtrips}. Thus we obtain e.g. that
for all $\Gamma:\Con_\Syn$, $\beta\,(\alpha\,\Gamma) = \Gamma$, and so on.

\section{Gluing for the global theory}
\label{sec:gluing}

In this section we extend the gluing proof of
\cite{DBLP:conf/rta/KaposiHS19} to the global theory. We first define
the notion of weak morphism which also has to respect $\foralle$, then
we define the gluing model. Finally we show how to make use of gluing
by defining a global section functor from the syntax of the global
theory $\e\Syn$ to the presheaf model $\PSh(\square)$ (Problem
\ref{con:psh_global}), and applying gluing to this, which shows that
our theory enjoys canonicity.

\subsection{Weak Morphism of Models Respecting $\foralle$}
\label{sec:weak_morph}

A weak morphism is a functor which has an action on types and terms
and preserves instantiation of substitution strictly and we require
that the empty context and context extension are preserved up to
isomorphism. Furthermore, we require that the functor preserves
$\foralle$, $\ke$, $\Re$ and $\Se$ strictly.

\begin{definition}[Weak morphisms of models of the global theory]
  Given two models of the global theory $\mathcal{C}$ and
  $\mathcal{D}$, a weak morphism $F$ from $\mathcal{C}$ to
  $\mathcal{D}$ contains the following functions and satisfies the
  following equations. We usually omit the subscripts saying which
  model is meant because it is clear from the context, and we also
  overload the names for the four maps in $F$.
  \begin{alignat*}{10}
    & \text{Maps} && \text{Preservation of core structure} && \text{Preservation of $\foralle,\ke,\Re,\Se$} \\
    & F : \Con_{\mathcal{C}} \ra \Con_{\mathcal{D}}                               && F\,(\sigma\circ\rho) = F\,\sigma\circ F\,\rho     && F\,(\foralle\,\Gamma) = \foralle(F\,\Gamma) \\
    & F : \Sub\,\Delta\,\Gamma \ra \Sub\,(F\,\Delta)\,(F\,\Gamma) \hspace{2em} && F\,\id = \id                                      && F\,(\foralle\,\sigma) = \foralle(F\,\sigma) \\
    & F : \Ty\,\Gamma \ra \Ty\,(F\,\Gamma)                                     && F_\epsilon^{-1} : \Sub\,\diamond\,(F\,\diamond)     && F\,(\foralle\,A) = \foralle(F\,A) \\
    & F : \Tm\,\Gamma\,A \ra \Tm\,(F\,\Gamma)\,(F\,A)                          && F_\epsilon^{-1}\circ\epsilon = \id_{F\,\diamond}      && F\,(\foralle\,t) = \foralle(F\,t) \\
    &                                                                          && F\,(A[\sigma]) = F\,A[F\,\sigma]                   && F\,\ke_\Gamma = \ke_{F\,\Gamma}               \\
    &                                                                          && F\,(t[\sigma]) = F\,t[F\,\sigma]                   && F\,\Re_\Gamma = \Re_{F\,\Gamma}               \\
    &                                                                          && F_{\ext}^{-1} : \Sub\,(F\,\Gamma\ext F\,A)\,(F\,(\Gamma\ext A)) \hspace{2em} && F\,\Se_\Gamma = \Se_{F\,\Gamma} \\
    &                                                                          && F_{\ext}^{-1}\circ(F\,\p,F\,\q) = \id                           \\
    &                                                                          && (F\,\p,F\,\q)\circ F_{\ext}^{-1} = \id                          
  \end{alignat*}
\end{definition}
For any such $F$ we can define comparison maps expressing that $\K$,
$\top$, $\Sigma$, $\Eq$ are preserved automatically up to isomorphism,
$\Pi$ and $\U$ are preserved in a lax way and $\Bool$ in an oplax
way. The specification is on the left, the implementations are on the
right.
\begin{alignat*}{10}
  & F_\K && {}: F\,(\K\,\Theta) \cong \K\,(F\,\Theta) && F_\K && {}:= F\,\q\circ F_{\ext}^{-1} \\
  & F_\Sigma && {}: F\,(\Sigma\,A\,B) \cong \Sigma\,(F\,A)\,(F\,B[F_{\ext}^{-1}]) && F_\Sigma && {}:= (F\,(\pi_1\,q), F\,(\pi_2\,q))[F_{\ext}^{-1}]\\
  & F_\Eq && {}: F\,(\Eq_A\,u\,v) \cong \Eq_{F\,A}\,(F\,u)\,(F\,v) && F_\Eq && {}:= \refl[F_{\ext}^{-1}] \\
  & F_\Pi && {}: \Tm\,(F\,\Gamma\ext F\,(\Pi\,A\,B)) (\Pi\,(F\,A)\,(F\,B[F_{\ext}^{-1}])[\p]) \hspace{2em} && F_\Pi && {}:= \lam\,(F\,(\app\,\q)[F_{\ext}^{-1}])[F_{\ext}^{-1}] \\
  & F_\U && {}: \Tm\,(F\,\Gamma\ext F\,\U)\,\U && F_\U && {}:= \c(F\,(\El\,\q)[F_{\ext}^{-1}]) \\
  & F_\Bool^{-1} && {}: \Tm\,(F\,\Gamma\ext\Bool)\,(F\,\Bool[\p]) && F_\Bool^{-1} && {}:= \ite\,\q\,(F\,\true[\p])\,(F\,\false[\p])
\end{alignat*}
We note that a consequence of preservation of $\foralle$ is that $\foralle F_{\ext}^{-1} = F_{\ext}^{-1}$.

\subsection{The Gluing Displayed Model}
\label{sec:dispmodel}

A displayed model over a base model is equivalent to a model with a
morphism into the base. It can also be seen as the collection of
motives and methods for the induction principle of the syntax (the initial
model). The components of a displayed model can be computed for any
generalised algebraic theory using the methods in
\cite{DBLP:journals/pacmpl/KaposiKA19}. We mark displayed components
with a bullet. For example, displayed contexts are a set
indexed over contexts: $\Con^\bullet : \Con\ra\Set$ where $\Con$ is
the set of contexts of the base model. Displayed types are indexed
implicitly over a base context, a displayed context at the base
context and a base type:
$\Ty^\bullet:\{\Gamma:\Con\}\ra\Con^\bullet\,\Gamma\ra\Ty\,\Gamma\ra\Set$. The
displayed variants of equations are well-typed because of the
corresponding equations in the base model.

\begin{problem}[Gluing]
  Given a weak morphism $F$ from $\mathcal{C}$ to $\mathcal{D}$, we
  construct a displayed model of the global theory over $\mathcal{C}$.
\end{problem}
\begin{proof}[Construction]
Gluing of a weak morphism from $\mathcal{C}$ to $\mathcal{D}$ gives a
displayed model over $\mathcal{C}$. This model construction is a
generalisation of logical predicates. A displayed context is a
predicate over $F$ applied to the context, a displayed type is a
dependent predicate over $F$ applied to the type. A displayed
substitution/term says expresses that $F$ of the substitution/term
respects the predicates.
\begin{alignat*}{10}
  & \Con^\bullet\,\Gamma && {} := \Ty\,(F\,\Gamma) && \Ty^\bullet\,\Gamma^\bullet\,A && {}:= \Ty\,(F\,\Gamma\ext\Gamma^\bullet\ext F\,A[\p]) \\
  & \Sub^\bullet\,\Delta^\bullet\,\Gamma^\bullet\,\sigma && {}:= \Tm\,(F\,\Delta\ext\Delta^\bullet)\,(\Gamma^\bullet[F\,\sigma\circ\p]) \hspace{3em} && \Tm^\bullet\,\Gamma^\bullet\,A^\bullet\,t && {}:= \Tm\,(F\,\Gamma\ext\Gamma^\bullet)\,(A^\bullet[\id,F\,t[\p]])
\end{alignat*}
The core calculus is glued using the same technique as the standard
model: composition and substitution are modelled by substitution,
context extension by $\Sigma$ types.
\begin{alignat*}{10}
  & \sigma^\bullet\circ^\bullet\rho^\bullet && {}:= \sigma^\bullet[F\,\rho\circ\p,\rho^\bullet]     && t^\bullet[\sigma^\bullet]^\bullet && {}:= t^\bullet[F\,\sigma\circ\p,\sigma^\bullet] \\                                           
  & \id^\bullet && {}:= \q                                                                          && \Gamma^\bullet\mathbin{{\ext}^\bullet} A^\bullet && {}:= \Sigma\,(\Gamma^\bullet[F\,\p])\,(A^\bullet[F\,\p\circ\p,\q,F\,\q[\p]]) \\
  & \diamond^\bullet && {}:= \top                                                                   && (\sigma^\bullet ,^\bullet t^\bullet) && {}:= (\sigma^\bullet , t^\bullet) \\                                                      
  & \epsilon^\bullet && {}:= \tt                                                                    && \p^\bullet && {}:= \pi_1\,\q \\                                                                                                   
  & A^\bullet[\sigma^\bullet]^\bullet && {}:= A^\bullet[(F\,\sigma\circ\p,\sigma^\bullet)\uparrow]  \hspace{7.5em} && \q^\bullet && {}:= \pi_2\,\q                                                                                                      
\end{alignat*}
The type formers $\K$, $\top$, $\Sigma$, $\Eq$, $\Pi$, $\U$ and
$\Bool$ in the gluing displayed model are defined as
follows. $\K^\bullet$ simply returns its argument context,
$\top^\bullet$, $\Sigma^\bullet$, $\Eq^\bullet$ are
pointwise. $\Pi^\bullet$ expresses that if the predicate holds for an
input, it holds for the output. $\U^\bullet$ gives predicate space,
$\Bool^\bullet$ for a $b$ in $F\,\Bool$ says that there is a boolean
which is equal to $b$ when mapped using $F_\Bool^{-1}$. The
definitions are straightforward, but technical: adjustments using
$F_{\ext}^{-1}$ and the comparison maps (e.g.\ $F_{\Sigma}$, $F_\Pi$)
have to make things match. The definitions below satisfy all the
equations of the displayed model.
\begin{alignat*}{10}
  & \K^\bullet\,\Theta^\bullet && {}:= \Theta^\bullet[F_\K][F\,\epsilon\circ\p^2,\q] \\
  & \top^\bullet && {}:= \top \\
  & \tt^\bullet && {}:= \tt \\
  & \Sigma^\bullet\,A^\bullet\,B^\bullet[\gamma,\gamma^\bullet,F_\Sigma^{-1}[\gamma,(a,b)]] &&{}:= \Sigma\,(A^\bullet[\gamma,\gamma^\bullet,a])\,(B^\bullet[F_{\ext}^{-1}\circ(\gamma,a)\circ\p, (\gamma^\bullet[\p],\q),b[\p]]) \\
  & (u^\bullet ,^\bullet v^\bullet) &&{}:= (u^\bullet , v^\bullet) \\
  & \pi_1^\bullet\,t^\bullet && {}:= \pi_1\,t^\bullet \\
  & \pi_2^\bullet\,t^\bullet && {}:= \pi_2\,t^\bullet \\
  & \Eq^\bullet_{A^\bullet}\,u^\bullet\,v^\bullet && {}:= \Eq_{A^\bullet[\p,F\, u[\p^2]]}\,(u^\bullet[\p])\,(v^\bullet[\p]) \\
  & \refl^\bullet && {}:= \refl
  \\
%\end{alignat*}
%\begin{alignat*}{10}
  & \Pi^\bullet\,A^\bullet\,B^\bullet[\gamma,\gamma^\bullet,t] &&{}:= \Pi\,(F\,A[\gamma])\,\big(\Pi\,(A^\bullet[(\gamma,\gamma^\bullet)\uparrow])\,\\
  & && \hspace{2em} (B^\bullet[F_{\ext}^{-1}\circ(\gamma\uparrow)\circ\p,(\gamma^\bullet[\p^2],\q),\app\,(F_\Pi[\gamma,t])[\p]])\big) \\
  & \lam^\bullet\,t^\bullet && {}:= \lam,\big(\lam\,(t^\bullet[F_{\ext}^{-1}\circ(\p^3,1),(2,0)])\big) \\
  & \app^\bullet\,t^\bullet && {}:= \app,(\app\,t)[F\,\p\circ\p,\pi_1\,q,F\,\q[\p],\pi_2\,\q] \\
  & \U^\bullet[\gamma,\gamma^\bullet,a] && {}:= \El\,(F_\U[\gamma,a]) \Ra \U \\
  & \El^\bullet\,a^\bullet && {}:= \El\,(\app\,a^\bullet) \\
  & \c^\bullet\,A^\bullet && {}:= \lam\,(\c\,A^\bullet) \\
  & \Bool^\bullet[\gamma,\gamma^\bullet,b] && {}:= \Sigma\,\Bool\,(\Eq_{F\,\Bool[\gamma][\p]}\,(F_\Bool^{-1}[\gamma\uparrow])\,(b[\p])) \\
  & \true^\bullet && {}:= (\true,\refl) \\
  & \false^\bullet && {}:= (\false,\refl) \\
  & \ite^\bullet\,t^\bullet\,u^\bullet\,v^\bullet && {}:= \ite\,(\pi_1\,t^\bullet)\,u^\bullet\,v^\bullet
\end{alignat*}
The displayed versions of the $\foralle$, $\ke$, $\Re$ and $\Se$
operators are defined as follows.
\[
\foralle^\bullet\Gamma^\bullet := \foralle\Gamma^\bullet \hspace{4em}
\foralle^\bullet \sigma^\bullet := \foralle\sigma^\bullet \hspace{4em}
\foralle^\bullet A^\bullet := \foralle A^\bullet \hspace{4em}
\foralle^\bullet t^\bullet := \foralle t^\bullet
\]
\[
{\ke^\bullet}_{\Gamma^\bullet} := \q[\ke_{F\,\Gamma\ext\Gamma^\bullet}] \hspace{4em}
{\Re^\bullet}_{\Gamma^\bullet} := \q[\Re_{F\,\Gamma\ext\Gamma^\bullet}] \hspace{4em}
{\Se^\bullet}_{\Gamma^\bullet} := \q[\Se_{F\,\Gamma\ext\Gamma^\bullet}]
\]
$\foralle^\bullet\Gamma^\bullet$ has to be in
$\Ty\,(F\,(\foralle\Gamma))$ and $\foralle\Gamma^\bullet$ is in
$\Ty\,(\foralle(F\,\Gamma))$, but these are equal because $F$ respects
$\foralle$ strictly. The situation is similar for the other operators.

For $\mkpie^\bullet\,\sigma^\bullet\,t_k^\bullet\,t^\bullet$, we need a term in
\begin{alignat*}{10}
  & \Tm\,(F\,\Delta\ext\Delta^\bullet)\,\big(\foralle(\Pi^\bullet\,A^\bullet\,B^\bullet)[F\,\sigma\circ\p,\sigma^\bullet,F\,(\mkpie\,\sigma\,t_k\,t)[\p]]\big) & = \\
  & & \hspace{-10em}\text{(introduce abbreviation $\rho$)} \\
  & \Tm\,(F\,\Delta\ext\Delta^\bullet)\,\big(\foralle(\Pi^\bullet\,A^\bullet\,B^\bullet)[\rho]\big) & = \\
  & \Tm\,(F\,\Delta\ext\Delta^\bullet)\,\bigg(\foralle\Big(\Pi\,(F\,A[\p^2])\,\big(\Pi\,(A^\bullet[\p^2,\q])\,(B^\bullet[F_{\ext}^{-1}\circ(\p^4,1),(3,0),F_{\Pi}[\p^4,2]\oldapp 1])\big)\Big)[\rho]\bigg) &\cong \\
  & & \hspace{-6em}\text{(curry-uncurry)} \\
  & \Tm\,(F\,\Delta\ext\Delta^\bullet)\,\bigg(\foralle\Big(\Pi\,\big(\Sigma\,(F\,A[\p^2])\,(A^\bullet[\p^2,\q])\big)\,\big(B^\bullet[F_{\ext}^{-1}\circ(\p^3,\pi_1\,0),(2,\pi_2\,0),F_{\Pi}[\p^3,1]\oldapp \pi_1\,0]\big)\Big)[\rho]\bigg).
\end{alignat*}
We can directly apply $\mkpie$ to build an element of this latter type
reusing the uncurried versions of $t_k^\bullet$ and $t^\bullet$, respectively.
\[
w := \mkpie\,\rho\,\Big(\lam\,\big(\app\,(\app\,t_k^\bullet)[\p,\pi_1\,\q,\pi_2\,\q]\big)\Big)\Big(\lam\,\big(\app\,(\app\,t^\bullet)[\p,\pi_1\,\q,\pi_2\,\q]\big)\Big)
\]
Now we curry the resulting term under $\foralle$ to obtain the definition of $\mkpie^\bullet$:
\[
\mkpie^\bullet\,\sigma^\bullet\,t_k^\bullet\,t^\bullet := \foralle(\lam\,(\lam\,(2\oldapp(1,0))))[\rho, w]
\]
For $\unspane^\bullet\,A_k^\bullet\,A^\bullet\,t_k^\bullet$, we need a term in
\begin{alignat*}{10}
  & \Tm\,(F\,\Gamma\ext\Gamma^\bullet)\,(\foralle^\bullet\,\U^\bullet[\epsilon^\bullet]^\bullet[\id,F\,(\unspan\,A_k\,A\,t_k)[\p]]) & = \\
  & \Tm\,(F\,\Gamma\ext\Gamma^\bullet)\,\Big(\foralle\big(F\,(\El\,\q)[F_{\ext}^{-1}]\Ra\U\big)\big[F\,\epsilon\circ\p,F\,(\unspan\,A_k\,A\,t_k)[\p]\big]\Big) & = \\
  & & \hspace{-3em}\text{(introduce abbreviation $\rho$)} \\
  & \Tm\,(F\,\Gamma\ext\Gamma^\bullet)\,\Big(\foralle\big(F\,(\El\,\q)[F_{\ext}^{-1}]\Ra\U\big)[\rho]\Big),
\end{alignat*}
which is the span of a predicate. We already have predicates on
$F\,A_k$ and $F\,A$ with a map from witnesses of the second one to the
first one that lies over $F\,t_k$:
\begin{alignat*}{10}
  & A_k^\bullet && {}: \Ty\,(F\,\Gamma\ext\Gamma^\bullet\ext F\,A_k[\p]) \\
  & A^\bullet && {}: \Ty\,(F\,\Gamma\ext\Gamma^\bullet\ext F\,A[\p]) \\
  & t_k^\bullet && {}: \Tm\,(F\,(\Gamma\ext A)\ext\Gamma^\bullet\mathbin{{\ext}^\bullet} A^\bullet)\,(A_k^\bullet[F\,\p\circ\p,\pi_1\,\q,F\,t_k[\p]])
\end{alignat*}
The way to make terms in $\foralle$ of a function space is by
$\mkpie$. The function (the predicate) at the base is essentially
$A_k^\bullet$, the predicate at the apex is given by
$\unspane\,(A_k^\bullet[\dots])\,A^\bullet\,(t_k^\bullet[\dots])$
where $A_k^\bullet$ and $t_k^\bullet$ have to be substituted to plug
in their dependencies properly.
\begin{alignat*}{10}
  & \unspane^\bullet\,A_k^\bullet\,A^\bullet\,t_k^\bullet := \\
  & \hspace{2em}\mkpie\,\rho\,\big(\lam\,(\c\,A_k^\bullet)\big)\,\bigg(\lam\,\Big(\unspane\, && \big(A_k^\bullet[\p, F\,t_k[F_{\ext}^{-1}][\p^2,\q]]\big)\,A^\bullet\,\big(t_k^\bullet[F_{\ext}^{-1}\circ(\p^3,1),(2,0)]\big)\Big)\bigg)
\end{alignat*}

The equations for the global syntax follow from the preservation
properties of the weak morphism. Here is a typical example:
\begin{alignat*}{10}
  & \foralle^\bullet\ke^\bullet_{\Gamma^\bullet}  & = \\
  & \q[\foralle\ke_{F\,\Gamma\ext\Gamma^\bullet}] & = \\
  & & \hspace{-5em}\text{(one of the five equations)} \\
  & \q[\ke_{\foralle(F\,\Gamma)\ext\foralle\Gamma^\bullet}\circ\Se_{F\,\Gamma\ext\Gamma^\bullet}] & = \\
  & \q[\ke_{\foralle(F\,\Gamma)\ext\foralle\Gamma^\bullet}][\Se_{F\,\Gamma\ext\Gamma^\bullet}] & = \\
  & & \hspace{-5em}\text{($F$ commutes with $\foralle$ on contexts)} \\
  & \q[\ke_{F\,(\foralle\Gamma)\ext\foralle\Gamma^\bullet}][\Se_{F\,\Gamma\ext\Gamma^\bullet}] & = \\
  & \q[\ke_{F\,(\foralle\Gamma)\ext\foralle\Gamma^\bullet}][\p\circ\Se_{F\,\Gamma\ext\Gamma^\bullet},\q[\Se_{F\,\Gamma\ext\Gamma^\bullet}]] & = \\
  & \q[\ke_{F\,(\foralle\Gamma)\ext\foralle\Gamma^\bullet}][\foralle(\foralle\p)\circ\Se_{F\,\Gamma\ext\Gamma^\bullet},\q[\Se_{F\,\Gamma\ext\Gamma^\bullet}]] & = \\
  & & \hspace{-5em}\text{(naturality of $\Se$)} \\
  & \q[\ke_{F\,(\foralle\Gamma)\ext\foralle\Gamma^\bullet}][\Se_{F\,\Gamma}\circ\foralle(\foralle\p),\q[\Se_{F\,\Gamma\ext\Gamma^\bullet}]] & = \\
  & \q[\ke_{F\,(\foralle\Gamma)\ext\foralle\Gamma^\bullet}][\Se_{F\,\Gamma}\circ\p,\q[\Se_{F\,\Gamma\ext\Gamma^\bullet}]] & = \\
  & & \hspace{-5em}\text{($F$ commutes with $\Se$)} \\
  & \q[\ke_{F\,(\foralle\Gamma)\ext\foralle\Gamma^\bullet}][F\,\Se_\Gamma\circ\p,\q[\Se_{F\,\Gamma\ext\Gamma^\bullet}]] & = \\
  & \ke^\bullet_{\forall^\bullet\Gamma^\bullet}\circ^\bullet\Se^\bullet_{\Gamma^\bullet} 
\end{alignat*}
This completes the construction of the gluing displayed model.
\end{proof}

\subsection{The Global Section Functor}
\label{sec:global_section}

We define a way to interpret morphisms in $\square$ in our syntax (see
Definition \ref{def:square}). The substitution $\ll f\rr_\Gamma$ for
an $f:\square(J,I)$ is defined mutually with three equations by
induction on $f$:
\begin{alignat*}{10}
  & \ll f\rr_\Gamma :{} && \Sub\,(\foralle^I\Gamma)\,(\foralle^J\Gamma) \\
  & && \cul f\cur_{\foralle^n\Gamma}\circ\foralle^{I+n}\ke_\Gamma = \foralle^{J+n}\ke_\Gamma\circ\cul f\cur_{\foralle^{1+n}\Gamma} \\
  & && \foralle^{J+n}\Re_\Gamma\circ\cul f\cur_{\foralle^n\Gamma} = \cul f\cur_{\foralle^{1+n}\Gamma}\circ\foralle^{I+n}\Re_\Gamma \\
  & && \foralle^{J+n}\Se_\Gamma\circ\cul f\cur_{\foralle^{2+n}\Gamma} = \cul f\cur_{\foralle^{2+n}\Gamma}\circ\foralle^{I+n}\Se_\Gamma
\end{alignat*}
$\foralle^I\,\Gamma$ denotes the $I$ times iteration of $\foralle$ on
$\Gamma$. The different cases\footnote{An alternative definition of
$\cul\blank\cur_\Gamma$ uses $\ll\suc\,f\rr_\Gamma = \foralle\ll
f\rr_\Gamma$, $\ll k_I\rr_\Gamma = \ke_{\foralle^I\Gamma}$, and so
on. This version generates a global section functor which only
satisfies $\G\,(\foralle\Gamma) \cong \foralle(\G\,\Gamma)$ up to
isomorphism. An isomorphism is enough to build a gluing model, however
it is much more tedious than our current, stricter approach.} of $\cul\blank\cur$:
\begin{alignat*}{10}
  & \ll f\circ g\rr_\Gamma && {}:= \ll g\rr_\Gamma\circ\ll f\rr_\Gamma \hspace{5em} && \ll\suc\,f\rr_\Gamma && {}:= \ll f\rr_{\foralle\Gamma} \hspace{8em} &&    \ll \R_I\rr_\Gamma && {}:= \foralle^I\Re_{\Gamma} \\ 
  & \ll\id\rr_\Gamma && {}:= \id &&                                                   \ll k_I\rr_\Gamma && {}:= \foralle^I\ke_{\Gamma} &&        \ll \S_I\rr_\Gamma && {}:= \foralle^I\Se_{\Gamma}    
\end{alignat*}
The fact that $\ll\blank\rr_\Gamma$ preserves the equations for the
morphisms in $\square$ is direct except for the three naturality
equations $k_I\circ f = \suc\,f\circ k_J$, $\R_I\circ\suc\,f =
f\circ\R_J$ and $\S_I\circ\suc\,(\suc\,f) = \suc\,(\suc\,f)\circ\S_J$
which follow from the $n=0$ cases of the above three equations.

We also prove by induction on $f$ that for any $\sigma$, $ \cul
f\cur_\sigma : \cul f\cur_\Gamma\circ\foralle^I\sigma =
\foralle^J\sigma\circ\cul f\cur_\Delta.  $

Now we construct the weak morphism on which we will glue for canonicity. This can be seen as a cubical nerve functor.
\begin{problem}[Global section functor]
  We construct a weak morphism $\G$ from $\e\Syn$ to $\PSh(\square)$.
\end{problem}
\begin{proof}[Construction]
\begin{alignat*}{10}
  & \G : \Con_{\e\Syn} \ra \Con_{\PSh(\square)} \hspace{8em}&&                         \G : \Sub_{\e\Syn}\,\Delta\,\Gamma \ra \Sub_{\PSh(\square)}\,(\G\,\Delta)\,(\G\,\Gamma) \\
  & \G\,\Gamma\,I := \Sub_{\e\Syn}\,\diamond\,(\foralle^I\,\Gamma) &&      \G\,\sigma\,\delta^I := \foralle^I\,\sigma\circ\delta^I                                 \\
  & \gamma^I[f]_{\G\,\Gamma} := \ll f\rr_\Gamma\circ\gamma^I 
\end{alignat*}
The functor laws for $\G\,\Gamma$ follow from the definition of
$\cul\blank\cur_\Gamma$ and from the categorical laws in
$\e\Syn$. Naturality of $\G\,\sigma$ is a consequence of $\cul
f\cur_\sigma$. On types and terms:
\begin{alignat*}{10}
  & \G : \Ty_{\e\Syn}\,\Gamma \ra \Ty_{\PSh(\square)}\,(\G\,\Gamma)   &&         \G : \Tm_{\e\Syn}\,\Gamma\,A \ra \Tm_{\PSh(\square)}\,(\G\,\Gamma)\,(\G\,A) \\
  & \G\,A\,I\,\gamma^I := \Tm_{\e\Syn}\,\diamond\,(\foralle^I\,A[\gamma^I]) \hspace{6em} &&   \G\,t\,\gamma^I := \foralle^I\,t[\gamma^I]                                    \\
  & a^I[f]_{\G\,A} := \q[\cul f\cur_{\Gamma\ext A}][\gamma^I,a^I] 
\end{alignat*}
% Functor laws for $\G\,A$:
% \begin{alignat*}{10}
%   & a^I[f\circ g]_{\G\,A} = \\
%   & \q[\cul f\circ g\cur_{\Gamma\ext A}][\gamma^I,a^I] = \\
%   & \q[\cul g\cur_{\Gamma\ext A}][\cul f\cur_{\Gamma\ext A}][\gamma^I,a^I] = \\
%   & \q[\cul g\cur_{\Gamma\ext A}][\p,\q][\cul f\cur_{\Gamma\ext A}][\gamma^I,a^I] = \\
%   & \q[\cul g\cur_{\Gamma\ext A}][\p\circ\cul f\cur_{\Gamma\ext A}\circ(\gamma^I,a^I),\q[\cul f\cur_{\Gamma\ext A}][\gamma^I,a^I]] = \\
%   & \q[\cul g\cur_{\Gamma\ext A}][\gamma^J,\q[\cul f\cur_{\Gamma\ext A}][\gamma^I,a^I]]   = \\
%   & a^I[f]_{\G\,A}[g]_{\G\,A}
% \end{alignat*}
% \[
% a^I[\id]_{\G\,A} = \q[\cul \id\cur_{\Gamma\ext A}][\gamma^I,a^I] = \q[\gamma^I,a^I] = a^I
% \]
The functor laws for $\G\,A$ follow directly and naturality for $\G\,t$ is proven as follows.
\begin{alignat*}{10}
  & \G\,t\,\gamma^I[f]_{\G\,A} = \q[\cul f\cur_{\Gamma\ext A}][\gamma^I,\foralle^I t[\gamma^I]] = \q[\cul f\cur_{\Gamma\ext A}][\foralle^I(\id,t)][\gamma^I] \overset{\cul f\cur_\sigma}{=} \q[\foralle^J(\id,t)][\cul f\cur_{\Gamma}][\gamma^I] = \\
  & \q[\id,\foralle^J t][\cul f\cur_{\Gamma}][\gamma^I] = \foralle^J t[\cul f\cur_\Gamma\circ\gamma^I] = \G\,t\,(\gamma^I[f]_{\G\,\Gamma})
\end{alignat*}
It is easy to see that $\G$ is a functor, it is more interesting that it preserves substitution of types:
\begin{alignat*}{10}
  & \G\,(A[\sigma])\,I\,\delta^I = \Tm_{\e\Syn}\,\diamond\,(\foralle^I (A[\sigma])[\delta^I]) = \Tm_{\e\Syn}\,\diamond\,(\foralle^I A[\foralle^I\sigma][\delta^I]) = \G\,A\,I\,(\foralle^I\sigma\circ\delta^I) = \\
  & \G\,A\,I\,(G\,\sigma\,\delta^I) = \G\,A[G\,\sigma]\,I\,\delta^I \\[-0.6em]
  & a^I[f]_{\G\,(A[\sigma])} = \q[\cul f\cur_{\Delta\ext A[\sigma]}][\delta^I,a^I] = \q[\foralle^J(\sigma\uparrow)][\cul f\cur_{\Delta\ext A[\sigma]}][\delta^I,a^I] \overset{\cul f\cur_{\sigma\uparrow}}{=} \\
  & \q[\cul f\cur_{\Gamma\ext A}][\foralle^I(\sigma\uparrow)][\delta^I,a^I] = \q[\cul f\cur_{\Gamma\ext A}][\foralle^I\sigma\circ\delta^I,a^I] = a^I[f]_{\G\,A[\G\,\sigma]}
\end{alignat*}
A similar proof shows that $\G$ preserves term substitution. We define
preservation of $\diamond$ and $\blank\ext\blank$ as follows. In the
latter case we replace a metatheoretic pairing operation with a
syntactic extended substitution former.
\begin{alignat*}{10}
  & \G_\diamond^{-1} : \Sub_{\PSh(\square)}\,\diamond\,(\G\,\diamond) \hspace{5em} && \G_{\ext}^{-1} : \Sub_{\PSh(\square)}\,(\G\,\Gamma\ext\G\,A)\,(\G\,(\Gamma\ext A)) \\
  & \G_\diamond^{-1}\,\_ := \epsilon_{\e\Syn} && \G_{\ext}^{-1}\,(\gamma^I,a^I) := (\gamma^I ,_{\ext_{\e\Syn}} a^I)
\end{alignat*}
They obviously satisfy the necessary equations using the universal
properties of $\diamond$ and $\blank\ext\blank$.

We verify that $\G$ preserves $\foralle$ and commutes with $\ke$,
$\Re$, $\Se$. $\G$ commutes with $\forall$ on contexts
\[
\G\,(\foralle\Gamma)\,I = \Sub_{\e\Syn}\,\diamond\,(\foralle^I(\foralle\Gamma)) = \Sub_{\e\Syn}\,\diamond\,(\foralle^{1+I}\Gamma) = \G\,\Gamma\,(1+I) = \foralle(\G\,\Gamma)\,I
\]
\[
\gamma^{1+I}[f]_{\G\,(\foralle\Gamma)} = \cul f\cur_{\foralle\Gamma}\circ\gamma^{1+I} = \cul\suc\,f\cur_{\Gamma}\circ\gamma^{1+I} = \gamma^{1+I}[\suc\,f]_{\G\,\Gamma} = \gamma^{1+I}[f]_{\foralle(\G\,\Gamma)},
\]
substitutions
\[
\G\,(\foralle\sigma)\,\delta^{1+I} = \foralle^I(\foralle\sigma)\circ\delta^{1+I} = \foralle^{1+I}\sigma\circ\delta^{1+I} = \G\,\sigma\,\delta^{1+I} = \foralle(\G\,\sigma)\,\delta^{1+I},
\]
types
\[
\G\,(\foralle A)\,I\,\gamma^{1+I} = \Tm\,\diamond\,(\foralle^{1+I} A[\gamma^{1+I}]) = \G\,A\,(1+I)\,\gamma^{1+I} = \foralle(\G\,A)\,I\,\gamma^{1+I}
\]
\[
a^{1+I}[f]_{\G\,(\foralle A)} = \q[\cul f\cur_{\foralle\Gamma\ext\foralle A}][\gamma^{1+I},a^{1+I}] = \q[\cul\suc\,f\cur_{\Gamma\ext A}][\gamma^{1+I},a^{1+I}] = a^{1+I}[f]_{\foralle(\G\, A)},
\]
and terms
\[
\G\,(\foralle t)\,\gamma^{1+I} = \foralle^I(\foralle t)[\gamma^{1+I}] = \foralle^{1+I} t[\gamma^{1+I}] = \G\,t\,\gamma^{1+I} = \foralle(\G\,t)\,\gamma^{1+I}.
\]
Preservation of $\ke$ is given by
\[
\G\,\ke_\Gamma\,\gamma^{1+I} = \foralle^{I}\ke_\Gamma\circ\gamma^{1+I} = \cul k_I\cur_\Gamma\circ\gamma^{1+I} = \gamma^{1+I}[\ke_I]_{\G\,\Gamma} = \ke_{\G\,\Gamma}\,\gamma^{1+I},
\]
preservation of $\Re$ and $\Se$ are shown analogously.
\end{proof}

\subsection{Reaping the Fruits}

\begin{theorem}[Canonicity]
  In the local or global syntax, for a $t:\Tm\,\diamond\,\Bool$, we have $t=\true$ or $t=\false$.
\end{theorem}
\begin{proof}
By gluing along the global section functor, we obtain a displayed
model over the syntax of the global theory, and given a $t :
\Tm_{\e\Syn}\,\diamond\,\Bool$, by induction of $\e\Syn$, we get a
\[
\Tm_{\PSh(\square)}\,(\G\,\diamond\ext\top)\,(\Sigma\,\Bool\,(\Eq\,(\ite\,\q\,(\G\,\true[\p])\,(\G\,\false[\p]))\,(\G\,t[\p]))),
\]
and providing the $\id_{\diamond} : \G\,\diamond\,\diamond$ and the element
of the metatheoretic unit set as input we obtain
\[
(b:\mathbbm{2}) \times \mathsf{if}\,b\,\mathsf{then}\,\true\,\mathsf{else}\,\false = t.
\]
Given a term in the local syntax $t : \Tm_\Syn\,\diamond\,\Bool$, we
have $\alpha\,t : \Tm_{\e\Syn}\,(\alpha\,\diamond)\,(\alpha\,\Bool)$
where $\alpha$ is the map from the global syntax to the local syntax
defined in Section \ref{sec:iso_summary}. As $\alpha$ does not affect
the core syntax this is $\alpha\,t : \Tm_{\e\Syn}\,\diamond\,\Bool$
and we learn from canonicity say that $\alpha\,t = \true$. Now
applying $\beta$ to such an equation we obtain $t = \beta\,(\alpha\,t)
= \beta\,\true = \true$ as $\beta$ also does not affect that core
calculus.
\end{proof}

\section{Future work}
\label{sec:future}

We presented a type theory with internal parametricity, a presheaf
model and a canonicity proof. It can be seen as a baby version of
higher observational type theory (HOTT). To obtain HOTT, we plan to
add the following additional features to our theory:
\begin{itemize}
\item a bridge type which can be seen as an indexed version of $\forall$,
\item Reedy fibrancy, which replaces spans by relations,
\item a strictification construction which turns the isomorphism for
  $\Pi$ types into a definitional equality (in case of bridge, we also
  need the same for $\Sigma$),
\item Kan fibrancy, which adds transport and turns the bridge type
  into a proper identity type. This would also change the
  correspondence between $\forall\U$ and spans into $\forall\U$ and
  equivalences.
\end{itemize}
We would also like to include general (higher) inductive and
coinductive types. Concerning the metatheory, we plan to use internal
language techniques
\cite{DBLP:phd/us/Sterling22,DBLP:conf/fscd/BocquetKS23} to obtain a
higher level canonicity proof, and extend it to normalisation.

Several of our constructions in this paper follow a generic pattern:
most of the global theory and gluing should be derivable from the
2-category in Figure \ref{fig:cubecat} similarly to the way it is done
for multimodal type theory \cite{DBLP:journals/corr/abs-2301-11842}.

\begin{acks}
We thank the anonymous reviewers for their comments and
suggestions. We thank Rafaël Bocquet, Hugo Herbelin, András Kovács
and Christian Sattler for discussions related to the topics of this
paper.

The first and third authors were supported by project no. \grantnum{hungary}{TKP2021-NVA-29} which has been implemented with the support provided by the \grantsponsor{hungary}{Ministry of Culture and Innovation of Hungary from the National Research, Development and Innovation Fund}{}, financed under the TKP2021-NVA funding scheme.

This material is based upon work supported by the \grantsponsor{airforce}{Air Force Office of Scientific Research}{} under award number \grantnum{airforce}{FA9550-21-1-0009}.
\end{acks}

\bibliography{b}

\end{document}